\documentclass[a4paper,USenglish]{lipics-v2018}

\nolinenumbers

\newcommand{\lv}[1]{#1}
\newcommand{\sv}[1]{}

\usepackage{boxedminipage}
\usepackage{cite}
\usepackage{color}
\usepackage[linesnumbered,boxed]{algorithm2e}
\usepackage{graphicx}
\usepackage{latexsym}
\usepackage{amsmath,verbatim}
\usepackage{amssymb}
\usepackage{enumerate}
\usepackage{todonotes}
\usepackage[draft,author=]{fixme}
\fxsetup{theme=color}
\usetikzlibrary{shapes}

\newtheorem{CLM}{Claim}
\theoremstyle{theorem}
\newtheorem{proposition}[theorem]{Proposition}
\newtheorem{ourfact}[theorem]{Fact}
\newtheorem{redrule}{Reduction Rule}
 \newtheorem{observation}{Observation}

\newcommand{\fes}{\text{\normalfont{\bfseries fes}}}

\newcommand{\tcw}{{\mathbf{tcw}}}

\newcommand{\adh}{{\mathbf{adh}}}
\newcommand{\tor}{{\mathbf{tor}}}

\newcommand{\PP}{\mathcal{P}}
\newcommand{\RR}{\mathcal{P}}

\newcommand{\QQ}{\mathcal{Q}}

\newcommand{\SSS}{\mathcal{S}}
\newcommand{\LLL}{\mathcal{L}}

\newcommand{\Nat}{\mathbb{N}}

\newcommand{\bigoh}{\mathcal{O}}

\newcommand{\cc}[1]{{\mbox{\textnormal{\textsf{#1}}}}\xspace}

\newcommand{\NP}{\cc{NP}}
\newcommand{\FPT}{\cc{FPT}}
\newcommand{\XP}{\cc{XP}}
\newcommand{\Weft}{{\cc{W}}}
\newcommand{\W}[1]{{\Weft}{\normalfont{[#1]}}}

\newcommand{\EDP}{\textsc{EDP}}
\newcommand{\VDP}{\textsc{VDP}}
\newcommand{\hy}{\hbox{-}\nobreak\hskip0pt}

\newcommand{\SB}{\{\,}%
\newcommand{\SM}{\;{|}\;}%
\newcommand{\SE}{\,\}}%

\newcommand{\pbDef}[3]{%
\noindent
\begin{center}
\begin{boxedminipage}{0.98 \columnwidth}
#1\\[5pt]
\begin{tabular}{l p{0.70 \columnwidth}}
Input: & #2\\
Question: & #3
\end{tabular}
\end{boxedminipage}
\end{center}
}

\newcommand{\pbDefP}[4]{%
\noindent
\begin{center}
\begin{boxedminipage}{0.98 \columnwidth}
#1\\[5pt]
\begin{tabular}{l p{0.70 \columnwidth}}
Input: & #2\\
Parameter: & #3\\
Question: & #4
\end{tabular}
\end{boxedminipage}
\end{center}
}

\usepackage{enumitem}

\sv{\title{The Power of Cut-Based Parameters for Computing Edge Disjoint Paths}}
\lv{\title{ The Power of Cut-Based Parameters for Computing Edge Disjoint Paths}}

\author{Robert Ganian}{Algorithms and Complexity group, TU Wien, Vienna, Austria}{ganian@ac.tuwien.ac.at}{}{}
\author{Sebastian Ordyniak}{Algorithms group, University of Sheffield, Sheffield, UK}{sordyniak@gmail.com}{}{}
\Copyright{Robert Ganian and Sebastian Ordyniak}

\begin{document}


\authorrunning{R. Ganian, S. Ordyniak}

\maketitle
\date{ }

\begin{abstract}
This paper revisits the classical Edge Disjoint Paths (EDP) problem,
  where one is given an undirected graph $G$ and a set of terminal
  pairs $P$ and asks whether $G$ contains a set of pairwise
  edge-disjoint paths connecting every terminal pair in $P$. 
Our aim is to identify structural properties (parameters) of graphs which
allow the efficient solution of EDP without restricting the placement of terminals in $P$ in any way. In this setting, EDP is known to remain NP-hard even on extremely restricted graph classes, such as graphs with a vertex cover of size $3$.

We present three results which use edge-separator based parameters to chart new islands of tractability in the complexity landscape of EDP. Our first and main result utilizes the fairly recent structural parameter treecut width (a parameter with fundamental ties to graph immersions and graph cuts): we obtain a polynomial-time algorithm for EDP on every graph class of bounded treecut width. Our second result shows that EDP parameterized by treecut width is unlikely to be fixed-parameter tractable. Our final, third result is a polynomial kernel for EDP parameterized by the size of a minimum feedback edge set in the graph.
\end{abstract}

\section{Introduction}
\label{sec:intro}
{\sc Edge Disjoint Paths} ({\sc EDP}) is a fundamental routing graph problem: we are given a graph $G$ and a set $P$ containing pairs of vertices (\emph{terminals}), and are asked to decide whether there is a set of $|P|$ pairwise edge disjoint paths in $G$ connecting each pair in $P$.
Similarly to its counterpart, the \textsc{Vertex Disjoint Paths} (\textsc{VDP}) problem, \textsc{EDP} has been at the center of numerous results in structural
graph theory, approximation algorithms, and parameterized
algorithms~\cite{RobertsonS95b,KawarabayashiKK14,ChekuriKS06,KolliopoulosS04,EneMPR16,ZhouTN00,NishizekiVZ01,GargVY97,FleszarMS16}.

Both \textsc{EDP} and \textsc{VDP} are \NP-complete in general~\cite{Karp75}, and a significant amount of research has focused on identifying structural properties which make these problems tractable.
For instance, Robertson and Seymour's seminal work in the Graph Minors project~\cite{RobertsonS95b} provides an $\bigoh(n^3)$ time algorithm for both problems for every fixed value of $|P|$. Such results are often viewed through the more refined lens of the \emph{parameterized complexity} paradigm~\cite{DowneyFellows13,CyganFKLMPPS14}; there, each problem is associated with a numerical parameter $k$ (capturing some structural property of the instance), and the goal is to obtain algorithms which are efficient when the parameter is small. Ideally, the aim is then to obtain so-called \emph{fixed-parameter} algorithms for the problem, i.e., algorithms which run in time $f(k)\cdot n^{\bigoh(1)}$ where $f$ is a computable function and $n$ the input size; the aforementioned result of Robertson and Seymour is hence an example of a fixed-parameter algorithm where $k=|P|$, and we say that the problem is \FPT (w.r.t.\ this particular parameterization). In cases where fixed-parameter algorithms are unlikely to exist, one can instead aim for so-called \XP\ algorithms, i.e., algorithms which run in polynomial time for every fixed value of $k$.

Naturally, one prominent question that arises is whether we can use the structure of the input graph  itself (captured via a \emph{structural parameter}) to solve \EDP\ and \VDP. Here, we find a stark contrast in the difficulty between these two, otherwise closely related, problems. Indeed, while \VDP\ is known to be \FPT\ with respect to the well-established structural parameter \emph{treewidth}~\cite{Scheffler94}, \EDP\ is \NP-hard even on graphs of treewidth $3$~\cite{FleszarMS16}. What's worse, the same reduction shows that \EDP\ remains \NP-hard even on graphs with a vertex cover of size $3$~\cite{FleszarMS16}, which rules out fixed-parameter and \XP\ algorithms for the vast majority of studied graph parameters (including, e.g., \emph{treedepth} and the \emph{size of a minimum feedback vertex set}). 

We note that previous research on the problem has found ways of circumventing these negative results by imposing additional restrictions. Zhou, Tamura and Nishizeki~\cite{ZhouTN00} introduced the notion of an augmented graph, which contains information about how terminal pairs need to be connected, and used the treewidth of this graph to solve \EDP.
Recent work~\cite{GanianOS17} has also observed that \EDP\ admits a fixed-parameter algorithm when parameterized by treewidth and the maximum degree of the graph.
\smallskip

\textbf{Our Contribution.}
The aim of this paper is to provide new algorithms and matching lower bounds for solving the \textsc{Edge Disjoint Paths} problem \emph{without imposing any restrictions on the number and placement of terminals}. In other words, our aim is to be able to identify structural properties of the graph which guarantee tractability of the problem without knowing any information about the placement of terminals. The only positive result known so far in this setting requires us to restrict the degree of the input graph; however, in the bounded-degree setting there is a simple treewidth-preserving reduction from \EDP\ to \VDP\ (see Proposition~\ref{pro:edpvdp}), and so the problem only becomes truly interesting when the input graphs can contain vertices of higher degree.

Our main result is an \XP\ algorithm for \EDP\ when parameterized by the structural parameter \emph{treecut width}~\cite{Wollan15,MarxWollan14}. Treecut width is inherently tied to the theory of graph immersions; in particular, it has a similar relationship to graph immersions and cuts as treewidth has to graph minors and separators. Since its introduction, treecut width has been successfully used to obtain fixed-parameter algorithms for problems which are unlikely to be \FPT\ w.r.t.\ treewidth~\cite{GanianKimSzeider15,GanianKO18}; however, this is the first time that it has been used to obtain an algorithm for a problem that is \NP-hard on graphs of bounded treewidth. 

One ``feature'' of algorithmically exploiting treecut width is that it requires the solution of a non-trivial dynamic programming step. In previous works, this was carried out mostly by direct translations into \textsc{Integer Linear Programming} instances with few integer variables~\cite{GanianKimSzeider15} or by using network flows~\cite{GanianKO18}. In the case of \EDP, the dynamic programming step requires us to solve an instance of \EDP\ with a vertex cover of size $k$ where every vertex outside of the vertex cover has a degree of $2$; we call this problem \textsc{Simple EDP} and solve it in the dedicated Section~\ref{sec:simple}. It is worth noting that there is only a very small gap between \textsc{Simple EDP} (for which we provide an \XP\ algorithm) and graphs with a vertex cover of size $3$ (where \EDP\ is known to be \NP-hard).

In view of our main result, it is natural to ask whether the algorithm can be improved to a fixed-parameter one. After all, given the parallels between \EDP\ parameterized by treecut width (an edge-separator based parameter) and \VDP\ parameterized by treewidth (a vertex-separator based parameter), one would rightfully expect that the fixed-parameter tractability result on the latter~\cite{Scheffler94} would be mirrored in the former case. Surprisingly, we rule this out by showing that \EDP\ parameterized by treecut width is \W{1}-hard~\cite{DowneyFellows13,CyganFKLMPPS14} and hence unlikely to be fixed-parameter tractable; in fact, we obtain this lower-bound result even in the more restrictive setting of \textsc{Simple EDP}. The proof is based on an involved reduction from an adapted variant of the \textsc{Multidimensional Subset Sum} problem~\cite{GanianOS17,GanianKO18} and forms our second main contribution.

Having ruled out fixed-parameter algorithms for \EDP\ parameterized by treecut width and in view of previous lower-bound results, one may ask whether it is even possible to obtain such an algorithm for any reasonable parameterization. We answer this question positively by using the size of a minimum feedback edge set as a parameter. In fact, we show an even stronger result: as our final contribution, we obtain a so-called \emph{linear kernel}~\cite{DowneyFellows13,CyganFKLMPPS14} for \EDP\ parameterized by the size of a minimum feedback edge set.
\smallskip

\textbf{Organization of the Paper.} After introducing the required preliminaries in Section~\ref{sec:pre}, we proceed to introducing \textsc{Simple EDP}, solving it via an \XP\ algorithm and establishing our lower-bound result (Section~\ref{sec:simple}). Section~\ref{sec:tcwalg} then contains our algorithm for \EDP\ parameterized by treecut width. Finally, in Section~\ref{sec:fes} we obtain a polynomial kernel for \EDP\ parameterized by the size of a minimum feedback edge set.

\sv{\smallskip \noindent {\emph{Statements whose proofs are located in the appendix are marked with $\star$.}}}

\section{Preliminaries}
\label{sec:pre}

We use standard terminology for graph theory, see for
instance~\cite{Diestel10}.
Given a graph $G$, we let $V(G)$ denote its vertex set and $E(G)$ its
edge set.
The (open) neighborhood of a vertex $x \in V(G)$ is the set $\{y\in V(G):xy\in E(G)\}$ and is denoted by $N_G(x)$. For a vertex subset $X$, the neighborhood of $X$ is defined as $\bigcup_{x\in X} N_G(x) \setminus X$ and denoted by $N_G(X)$; we drop the subscript if the graph is clear from the context. \emph{Contracting} an edge $\{a,b\}$ is the operation of replacing vertices $a,b$ by a new vertex whose neighborhood is $(N(a)\cup N(b))\setminus \{a,b\}$.
For a vertex set $A$ (or edge set $B$), we use $G-A$ ($G-B$) to denote the graph obtained from
$G$ by deleting all vertices in $A$ (edges in $B$), and we use $G[A]$ to denote the
\emph{subgraph induced on} $A$, i.e., $G- (V(G)\setminus A)$. 

A \emph{forest} is a graph without cycles, and an edge set $X$ is a
\emph{feedback edge set} if $G-X$ is a forest. The \emph{feedback edge set
  number} of a graph $G$, denoted by $\fes(G)$, is the smallest
integer $k$ such that $G$ has a feedback edge set of size $k$.
We use $[i]$ to
denote the set $\{0,1,\dots,i\}$.

\subsection{Parameterized Complexity}
\lv{
A \emph{parameterized problem} $\PP$ is a subset of $\Sigma^* \times \Nat$ for some finite alphabet $\Sigma$. Let $L\subseteq \Sigma^*$ be a classical decision problem for a finite alphabet, and let $p$ be a non-negative integer-valued function defined on $\Sigma^*$. Then $L$ \emph{parameterized by} $p$ denotes the parameterized problem $\SB(x,p(x)) \SM x\in L \SE$ where $x\in \Sigma^*$.  For a problem instance $(x,k) \in \Sigma^* \times \Nat$ we call $x$ the main part and $k$ the parameter.  
A parameterized problem $\PP$ is \emph{fixed-parameter   tractable} (FPT in short) if a given instance $(x, k)$ can be solved in time  $\bigoh(f(k) \cdot p(|x|))$ where $f$ is an arbitrary computable function of $k$ and $p$ is a polynomial function; we call algorithms running in this time \emph{fixed-parameter algorithms}.

  Parameterized complexity classes are defined with respect to {\em fpt-reducibility}. A parameterized problem $P$ is {\em fpt-reducible} to $Q$ if in time $f(k)\cdot |x|^{O(1)}$, one can transform an instance $(x,k)$ of $\PP$ into an instance $(x',k')$ of $\QQ$ such that $(x,k)\in \PP$ if and only if $(x',k')\in \QQ$, and $k'\leq g(k)$, where $f$ and $g$ are computable functions depending only on $k$. 
 Owing to the definition, if $\PP$ fpt-reduces to $\QQ$
 and $\QQ$ is fixed-parameter tractable then $P$ is fixed-parameter
 tractable as well. 
 Central to parameterized complexity is the following hierarchy of complexity classes, defined by the closure of canonical problems under fpt-reductions:
 \[\FPT \subseteq \W{1} \subseteq \W{2} \subseteq \cdots \subseteq \XP.\] All inclusions are believed to be strict. In particular, $\FPT\neq \W{1}$ under the Exponential Time Hypothesis. 

 A major goal in parameterized complexity is to distinguish between parameterized problems which are in $\FPT$ and those which are $\W{1}$-hard, i.e., those to which every problem in $\W{1}$ is fpt-reducible. There are many problems shown to be complete for $\W{1}$, or equivalently $\W{1}$-complete, including the {\sc Multi-Colored Clique (MCC)} problem~\cite{DowneyFellows13}. We refer the reader to the respective monographs~\cite{FlumGrohe06,DowneyFellows13,CyganFKLMPPS15} for an in-depth
introduction to parameterized complexity.
 }
  \sv{
The parameterized complexity paradigm~\cite{FlumGrohe06,DowneyFellows13,CyganFKLMPPS15} allows a finer analysis of the complexity of problems by associating each problem instance $L$ with a numerical parameter $k$; the pair $(L,k)$ is then an instance of a parameterized problem. A parameterized problem is \emph{fixed-parameter   tractable} (\FPT\ in short) if a given instance $(L, k)$ can be solved in time  $f(k) \cdot |L|^{\bigoh(1)}$ where $f$ is an arbitrary computable function; we call algorithms running in this time \emph{fixed-parameter algorithms}. The complexity class \W{1} is often used to rule out fixed-parameter algorithms for a parameterized problem: under established complexity assumptions (including the \emph{Exponential Time Hypothesus}~\cite{ImpagliazzoPZ01}), problems that are hard for $\W{1}$ do not admit fixed-parameter algorithms. A parameterized problem $\PP$ is said to admit a \emph{linear kernel} if there exists a polynomial-time algorithm which converts an instance $(G,K)$ of $\PP$ into an equivalent instance $(G',k')$ such that $|G'|+k'\in \bigoh(k)$.
  }
  
  \subsection{Edge Disjoint Path Problem}\label{ssec:edp}

  Throughout the paper we consider the following problem.
\pbDef{\textsc{Edge Disjoint Paths (EDP)}}
{A graph $G$ and a set $P$ of \emph{terminal pairs}, i.e., a set of subsets of $V(G)$ of size two.}
{Is there a set of pairwise edge disjoint paths connecting every set
  of terminal pairs in $P$?}

  A vertex which occurs in a terminal pair is called a \emph{terminal}, and a set of pairwise edge disjoint paths connecting every set of terminal pairs in $P$ is called a \emph{solution}. Without loss of generality, we assume that $G$ is connected. The \textsc{Vertex Disjoint Paths (VDP)} problem is defined analogously as \EDP, with the sole distinction being that the paths must be vertex-disjoint.

The following proposition establishes a link between \EDP\ and \VDP\ on graphs of bounded degree. Since we will not use the notion of \emph{treewidth}~\cite{RobertsonS03b} anywhere else in the paper, we refer to the standard textbooks~\cite{DowneyFellows13,CyganFKLMPPS15} for its definition. 

\lv{\begin{proposition}}
\sv{\begin{proposition}[$\star$]}
\label{pro:edpvdp}
There exists a linear-time reduction from \EDP\ to \VDP\ with the following property: if the input graph has treewidth $k$ and maximum degree $d$, then the output graph has treewidth at most $k\cdot d+1$.
\end{proposition}

\lv{
\begin{proof}
Let $(G,P)$ be an instance of \EDP\ where $G$ has treewidth $k$ and maximum degree $d$; let $V=V(G)$ and $E=E(G)$. Observe that if any vertex $v\in V$ occurs in $P$ more than $d$ many times, then $(G,P)$ must be a \textbf{NO}-instance (we assume that $P$ does not contain tuples in the form $(a,a)$ for any $a$). 

Consider the graph $G'$ obtained in the following two-step procedure. First, we subdivide each edge in $G$ (i.e., we replace that edge with a vertex of degree $2$ that is adjacent to both endpoints of the original edge); let $V'$ be the set of vertices created by such subdivisions. Second, for each vertex $v=v_1\in V$ of the original graph $G$, we create $d-1$ copies $v_2,\dots, v_d$ of that vertex and set their neighborhood to match that of $v_1$. This construction gives rise to a natural mapping $\alpha$ from $G$ to $G'$ which maps each $v\in V$ to the set $v_1,\dots,v_d$ and each $e\in E$ to the vertex created by subdividing $e$. Next, we iteratively process $P$ as follows: for each $\{v,w\}\in P$, we add a tuple $\{v', w'\}$ into the set $P'$ such that $v'\in \alpha(v)$, $w'\in \alpha(w)$ and neither $v'$ nor $w'$ occurs in any other pair in $P'$ (the last condition can be ensured because each vertex in $v$ has $d$ copies in $G'$ but never occurs more than $d$ times in $P$).

It is now easy to verify that $(G,P)$ is a \textbf{YES}-instance of \EDP\ if and only if $(G',P')$ is a \textbf{YES}-instance of \VDP; indeed, such solutions can be converted to each other by applying $\alpha$ on each path, whereas for the forward direction we simply need to make sure that each path that passes through a vertex $v\in V$ uses a new vertex from $\alpha(v)$. Finally, one can convert any tree-decomposition $(T,X)$~\cite{DowneyFellows13} of $G$ of width $k$ into a tree-decomposition of $G'$ of width $k\cdot d+1$ by (1) replacing each vertex $v$ by $\alpha(v)$ in $T$, and then (2) by choosing, for each edge $e=ab\in E$, a bag $X\supseteq \{a,b\}$, creating a bag $X'=X\cup \{\alpha(e)\}$, and attaching $X'$ to $X$ as a leaf.
\end{proof}
}

We remark that Proposition~\ref{pro:edpvdp} in combination with the known fixed-parameter algorithm for \VDP\ parameterized by treewidth~\cite{Scheffler94} provides an alternative proof for the fixed-parameter tractability of \EDP\ parameterized by degree and treewidth~\cite{GanianOS17}. 

\subsection{Treecut Width}
\label{sub:tcw}
The notion of treecut decompositions was introduced by Wollan~\cite{Wollan15}, see also~\cite{MarxWollan14}.
A family of subsets $X_1, \ldots, X_{k}$ of $X$ is a {\em near-partition} of $X$ if they are pairwise disjoint and $\bigcup_{i=1}^{k} X_i=X$, allowing the possibility of $X_i=\emptyset$.  

\begin{definition}
A {\em treecut decomposition} of $G$ is a pair $(T,\mathcal{X})$ which consists of a rooted tree $T$ and a near-partition $\mathcal{X}=\{X_t\subseteq V(G): t\in V(T)\}$ of $V(G)$. A set in the family $\mathcal{X}$ is called a {\em bag} of the treecut decomposition. 
\end{definition}

For any node $t$ of $T$ other than the root $r$, let $e(t)=ut$ be the unique edge incident to $t$ on the path to $r$. Let $T_u$ and $T_t$ be the two connected components in $T-e(t)$ which contain $u$ and $t$, respectively. Note that $(\bigcup_{q\in T_u} X_q, \bigcup_{q\in T_t} X_q)$ is a near-partition of $V(G)$, and we use $E_t$ to denote the set of edges with one endpoint in each part. We define the {\em adhesion} of $t$ ($\adh(t)$) as $|E_t|$; if $t$ is the root, we set $\adh(t)=0$ and $E(t)=\emptyset$.

The {\em torso} of a treecut decomposition $(T,\mathcal{X})$ at a node $t$, written as $H_t$, is the graph obtained from $G$ as follows. If $T$ consists of a single node $t$, then the torso of $(T,\mathcal{X})$ at $t$ is $G$. Otherwise let $T_1, \ldots , T_{\ell}$ be the connected components of $T-t$. For each $i=1,\ldots , \ell$, the vertex set $Z_i\subseteq V(G)$ is defined as the set $\bigcup_{b\in V(T_i)}X_b$. The torso $H_t$ at $t$ is obtained from $G$ by {\em consolidating} each vertex set $Z_i$ into a single vertex $z_i$ (this is also called \emph{shrinking} in the literature). Here, the operation of consolidating a vertex set $Z$ into $z$ is to substitute $Z$ by $z$ in $G$, and for each edge $e$ between $Z$ and $v\in V(G)\setminus Z$, adding an edge $zv$ in the new graph. We note that this may create parallel edges.

The operation of {\em suppressing} (also called \emph{dissolving} in the literature) a vertex $v$ of degree at most $2$ consists of deleting~$v$, and when the degree is two, adding an edge between the neighbors of $v$. Given a connected graph $G$ and  $X\subseteq V(G)$, let the {\em 3-center} of $(G,X)$ be the unique graph obtained from $G$ by exhaustively suppressing vertices in $V(G) \setminus X$ of degree at most two. Finally, for a node $t$ of $T$, we denote by $\tilde{H}_t$ the 3-center of $(H_t,X_t)$, where $H_t$ is the torso of $(T,\mathcal{X})$ at $t$. 
Let the \emph{torso-size} $\tor(t)$ denote $|\tilde{H}_t|$. 

\begin{definition}
The width of a treecut decomposition $(T,\mathcal{X})$ of $G$ is $\max_{t\in V(T)}\{ \adh(t), \tor(t) \}$. The treecut width of $G$, or $\tcw(G)$ in short, is the minimum width of $(T,\mathcal{X})$ over all treecut decompositions $(T,\mathcal{X})$ of $G$.
\end{definition}

We conclude this subsection with some notation related to treecut decompositions. 
Given a tree node $t$, let $T_t$ be the subtree of $T$ rooted at $t$. Let $Y_t=\bigcup_{b\in V(T_t)} X_b$, and let $G_t$  denote the induced subgraph $G[Y_t]$. 
A node $t\neq r$ in a rooted treecut decomposition is \emph{thin} if $\adh(t)\leq 2$ and \emph{bold} otherwise.

\begin{figure}[ht]
\begin{center}
\begin{tikzpicture}
\tikzstyle{every node}=[draw, shape=circle, minimum size=3pt,inner sep=2pt, fill=white]

\draw (0,0) node[label=left:$a$] (a){}; 
\draw (1,0) node [label=270:$d$] (d){};
\draw (0,1) node[label=left:$b$] (b){}; 
\draw (1,1) node [label=right:$c$] (c){};
\draw (2,0) node [label=right:$e$](e){}; \draw (2,0.5) node [label=right:$f$](f){};
\draw (2,1) node [label=right:$g$](g){};
\draw (b)--(d)--(a)--(b)--(c)--(d)--(e)--(f); \draw (f)--(d)--(g);
\end{tikzpicture}
\quad\quad
\begin{tikzpicture}[scale=0.9]
\tikzstyle{cloud} = [draw, ellipse, minimum height=0em, minimum width=3em]

\node (d) at (0,2)[cloud,label=left:{$(2,0)$}] {d};
\node (a) at (-1,1)[cloud,label=left:{$(3,3)$}] {a};
\node (bc) at (-1,0)[cloud,label=left:{$(3,3)$}] {bc};
\node (e) at (0.5,1)[cloud,label=270:{$(1,2)$}] {e};
\node (f) at (2,1)[cloud,label=270:{$(1,2)$}] {f};
\node (g) at (3.5,1)[cloud,label=270:{$(1,1)$}] {g};

\draw[ultra thick] (bc)--(a)--(d);
\draw (d)--(e); \draw (f)--(d)--(g);

\tikzstyle{empty}=[draw, shape=circle, minimum size=3pt, inner sep=0pt, fill=white, color=white]


\end{tikzpicture}
\end{center}
\vspace{-0.5cm}
\caption{A graph $G$ and a width-$3$ treecut decomposition of $G$, including the torso-size (left value) and adhesion (right value) of each node.}
\vspace{-0.2cm}
\end{figure}
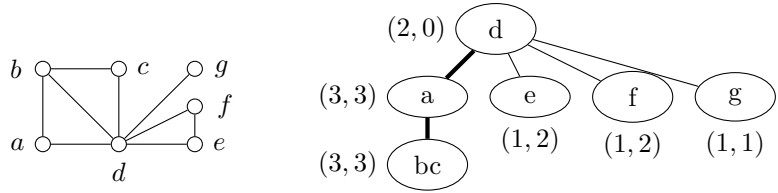

While it is not known how to compute optimal treecut decompositions efficiently, there exists a fixed-parameter 2-approximation algorithm which fully suffices for our purposes.

\begin{theorem}[\cite{KimOPST15}]
\label{thm:computetcdec}
There exists an algorithm that takes as input an $n$-vertex graph $G$ and integer $k$, runs in time $2^{\bigoh(k^2 \log k)} n^2$, and either outputs a treecut decomposition of $G$ of width at most $2k$ or correctly reports that $\tcw(G)> k$.
\end{theorem}

\sv{
A treecut decomposition $(T,\mathcal{X})$ is \emph{nice} if it satisfies the following condition for every thin node $t\in V(T)$: $N(Y_t)\cap (\bigcup_{b\text{ is a sibling of }t}Y_b)=\emptyset$.
The intuition behind nice treecut decompositions is that we restrict the neighborhood of thin nodes in a way which facilitates dynamic programming. Every treecut decomposition can be transformed into a nice treecut decomposition of the same width in cubic time~\cite{GanianKimSzeider15}.

For a node $t$, we let $B_t=\SB b\text{ is a child of }t\SM |N(Y_b)|\leq 2\wedge N(Y_b)\subseteq X_t \SE$ denote the set of thin children of $t$ whose neighborhood is a subset of $X_t$, and we let $A_t=\SB a\text{ is a child of }t\SM a\not \in B_t \SE$ be the set of all other children of $t$.
Then $|A_t|\leq 2k+1$ for every node $t$ in a nice treecut decomposition~\cite{GanianKimSzeider15}.
}
\lv{
A treecut decomposition $(T,\mathcal{X})$ is \emph{nice} if it satisfies the following condition for every thin node $t\in V(T)$: $N(Y_t)\cap (\bigcup_{b\text{ is a sibling of }t}Y_b)=\emptyset$.
The intuition behind nice treecut decompositions is that we restrict the neighborhood of thin nodes in a way which facilitates dynamic programming. 

\begin{lemma}[\cite{GanianKimSzeider15}]
\label{lem:nicetcdec}
There exists a cubic-time algorithm which transforms any rooted treecut decomposition $(T,\mathcal{X})$ of $G$ into a nice treecut decomposition of the same graph, without increasing its width or number of nodes.
\end{lemma}

For a node $t$, we let $B_t=\SB b\text{ is a child of }t\SM |N(Y_b)|\leq 2\wedge N(Y_b)\subseteq X_t \SE$ denote the set of thin children of $t$ whose neighborhood is a subset of $X_t$, and we let $A_t=\SB a\text{ is a child of }t\SM a\not \in B_t \SE$ be the set of all other children of $t$.
The following property of nice treecut decompositions will be crucial for our algorithm.

\begin{lemma}[\cite{GanianKimSzeider15}]
\label{lem:Asmall}
Let $t$ be a node in a nice treecut decomposition of width $k$. Then $|A_t|\leq 2k+1$.
\end{lemma}
}

We refer to previous work~\cite{MarxWollan14,KimOPST15,GanianKimSzeider15} for a comparison of treecut width to other parameters. Without loss of generality, we shall assume that $X_r=\emptyset$.

\section{The Simple Edge Disjoint Paths Problem}
\label{sec:simple}

Before we start working towards our algorithm for solving \EDP\ parameterized by treecut width, we will first deal with a simpler (but crucial) setting for the problem. We call this the \textsc{Simple Edge Disjoint Paths} problem (\textsc{Simple EDP)} and define it below.

\pbDefP{
\textsc{Simple EDP}}
{An \textsc{EDP} instance $(G,P)$ such that $V(G)=A\cup B$ where $B$ is an independent set containing vertices of degree at most $2$.}
{$k=|A|$}
{Is $(G,P)$ a \textbf{YES}-instance of \textsc{EDP}?}

Notice that every instance of \textsc{Simple EDP} has treecut width at most $k$, and so it forms a special case of \EDP\ parameterized by treecut width. Indeed, the treecut decomposition where $T$ is a star, the center bag contains $A$, and each leaf bag contains a vertex from $B$ (except for the root $r$, where $X_r=\emptyset$), has treecut width at most $k$. This contrasts to the setting where $G$ has a vertex cover of size $3$ and all vertices outside the vertex cover have degree $3$; the treecut width of such graphs is not bounded by any constant, and \EDP\ is known to be \NP-complete in this setting~\cite{FleszarMS16}.

The main reason we introduce and focus on \textsc{Simple EDP} is that it captures the combinatorial problem that needs to be solved in the dynamic step of the algorithm for \EDP\ parameterized by treecut width. Hence, our first task here will be to solve \textsc{Simple EDP} by an algorithm that can later be called as a subroutine.

\lv{\begin{lemma}}
\sv{\begin{lemma}[$\star$]}
\label{lem:simplealg}
\textsc{Simple EDP} can be solved in time $\bigoh(|P|^{\binom{k}{2}+1}(k+1)!)$.
\end{lemma}

\lv{\begin{proof}}
\sv{\begin{proof}[Sketch of Proof.]}
  Let $(G,P)$ with partition $A$ and $B$ and $k=|A|$ be an instance of
  \textsc{Simple EDP}. 
  Let the \emph{terminal graph} of $G$, denoted by $G^T$, as the graph with vertex set $V$ and edge set $P$.
  \sv{Our first course of action will be to simplify the instance by removing all vertices in $B$ that are not part of any terminal pair; to this end, we add multi-edges into $G[A]$ which represent removed degree-2 vertices. We now make the following two observations:
 } 
\lv{

  We will start by simplifying
  the instance using some simple observations. First we will show that
  we can remove all vertices in $B$ that are not contained in any
  terminal pair by adding multi-edges to $G[A]$. Namely, let $v$ be a
  vertex in $B$ that does not appear in any terminal pair in
  $P$. If $v$ has no neighbors or at most one neighbor, then
  $v$ can simply be removed from $G$, and if $v$ has degree two, then
  we can remove $v$ and add an edge between its two neighbors in $A$.
  Hence in the following we will assume that all vertices in $B$ occur
  in at least one terminal pair and that $G[A]$ can contain
  multi-edges.
  
  The following two observations will be crucial for our algorithm:
  }  
  \begin{itemize}
  \item[O1] 
    Consider a path $P$ connecting a terminal pair $p \in P$ in a
    solution. Because $B$ is an independent set and every vertex in $B$
    has degree at most two and is contained in at least one terminal pair in $P$, we obtain that
    all inner vertices of $P$ are from $A$. Hence, $P$ contains at most $k+2$ vertices and all inner vertices of $P$ are
    contained in $A$. It follows that $P$ is completely characterized
    by the sequence of vertices it uses in
    $A$. Consequently,
    there are at most $\sum_{\ell=1}^k\binom{k}{\ell}\ell!\leq (k+1)!$ different
    types of paths that
    need to be considered for the connection of any terminal pair.
  \item[O2] $G^T[B]$ is a disjoint union of paths and
    cycles. This is because every vertex $v$ of $G$ can be contained in
    at most $|N_G(v)|$ terminal pairs in $P$ (otherwise we immediately reject) and all vertices in $B$ have degree at most two.
  \end{itemize}
  Let $u$ and $v$ be two distinct vertices in $A$. Because $|A|\leq k$, we can
  enumerate all possible paths between $u$ and $v$ in $G[A]$ in time
  $\bigoh((k+1)!)$. We will represent each such path $H$ as a binary
  vector $E_H$, whose entries are indexed by all sets of two distinct
  vertices in $A$,
  such that $E_H[e]=1$ if $H$ uses the edge $e$ and $E_H[e]=0$
  otherwise. Moreover, we will denote by $E_{u,v}$ the set $\SB E_H
  \SM H\textup{ is a path between }u\textup{ and }v\textup{ in
  }G[A]\SE$; intuitively, $E_{u,v}$ captures all possible sets of edges that need to be used in order to connect $u$ to $v$.

  Let $S$ be a solution for $(G,P)$. The algorithm represents every
  solution $S$ for $(G,P)$ as a solution vector $E_S$ of natural numbers whose
  entries are indexed by all sets $\{u,v\}$ of two distinct vertices
  in $A$. More specifically, for two distinct vertices $u$ and $v$ in
  $A$, $E_S[\{u,v\}]$ is equal to the number of edges between $u$ and
  $v$ used by the paths in $S$. The algorithm uses dynamic programming to
  compute the set $\LLL$ of all solution vectors; clearly, $\LLL\neq \emptyset$ if and only if $(G,P)$ is a \textbf{YES}-instance. We compute $\LLL$ in two main steps:
  \begin{itemize}
  \item[(S1)] the algorithm computes the set $\LLL_A$ of all solution vectors
    for the sub-instance $(G[A],P')$ of $(G,P)$, where $P'$ is the
    subset of $P$ containing all terminal pairs $\{p,q\}$ with $p,q \in
    A$.
  \item[(S2)] the algorithm computes the set of all solution vectors
    for the sub-instance $(G,P \setminus P')$. Note that every
    terminal pair $p$ in $P \setminus P'$ is either completely
    contained in $B$, in which case it forms an edge of a path or
    acycle in $G^T[B]$, or $p$ has one vertex in $A$ and the other
    vertex in $B$, which is the endpoint of a path in $G^T[B]$. The
    algorithm now computes the set of all solution vectors for the 
    sub-instance $(G,P \setminus P')$ in two steps:
    \begin{itemize}
    \item[(S2A)] For every cycle $C$ in $G^T[B]$, the algorithm
      computes the set $\LLL_C$ of all solution vectors for the sub-instance $(G[A
      \cup V(C)], P_C)$, where $P_C$ is the subset of $P$ containing all terminal pairs $\{p,q\}$ such that $p,q \in C$. 
    \item[(S2B)] For every path $H$ in $G^T[B]$, the algorithm
      computes the set $\LLL_H$ of all solution vectors for the sub-instance $(G[A\cup
      V(H)],P_H)$, where $P_H$ is the subset of $P$ containing all terminal
      pairs $\{p,q\}$ with $\{p,q\} \cap V(H)\neq \emptyset$.
    \end{itemize}
  \end{itemize}
  In the end, the set of all \emph{hypothetical solution vectors} $\LLL'$ for $(G,P)$ is
  obtained as $\LLL_A\oplus (\oplus_{C \textup{ is a cycle of
    }G^T[B]}\LLL_C)\oplus (\oplus_{H \textup{ is a path of
    }G^T[B]}\LLL_H)$, where $\RR\oplus \RR'$ for two sets $\RR$ and
  $\RR'$ of
  solution vectors is equal to $\SB R+R' \SM R \in \RR \land R'\in
  \RR'\SE$. Each vector in $\LLL'$ describes one possible set of multi-edges in $G[A]$ that can be used to connect all terminal pairs in $P$. In order to compute $\LLL$, one simply needs to remove all vectors from $\LLL'$ which require more multi-edges than are available in $G[A]$; in particular, to obtain $\LLL$ we delete each $S$ from $\LLL'$ such that there exist $u,v\in A$ where $E_S[\{u,v\}]$ exceeds the number of multi-edges between $u$ and $v$ in $G$. The
  algorithm then returns \textbf{YES} if $\LLL$ is non-empty and otherwise the algorithm returns \textbf{NO}. 
  \lv{  Note that, as is usually the case with these types of dynamic
  programming algorithms, the algorithm can also be easily modified to find
  a solution for $(G,P)$, without increasing its running time.

      The set $\LLL_A$ described in step (S1) is computed as follows.
  Given an arbitrary but
  fixed ordering $p_1,\dotsc,p_{|P'|}$ of the terminal pairs in $P'$, let $P_i$ be the set
  $\SB p_j \SM 1\leq j \leq i\SE$, for every $i$ with $1 \leq
  i \leq |P'|$. The algorithm now uses dynamic programming to compute
  the sets $S_1,\dotsc,S_{|P'|}$, where $S_i$ contains the set of
  all hypothetical solution vectors for the instance $(G[A],P_i)$ as follows.
  The algorithm starts by setting
  $T_1$ to be the set $E_{p_1}$. Then for every $i$ with
  $1< i\leq |P'|$, the algorithm computes $T_i$ from $T_{i-1}$ as the set
  $\SB E+E' \SM E \in T_{i-1}\land E' \in E_{p_i}\SE$.

  The set $\LLL_C$ described in step (S2A) for a cycle $C=(v_1,\dotsc,v_n)$ of $G^T[B]$
  is computed as follows.
  The algorithm starts by computing a table
  $T_i$ for every $2 \leq i \leq n$, which for every $n_1 \in N_G(v_1)$
  and $n_2 \in N_G(v_2)$
  contains the set of all solution vectors for the instance $(G[A\cup
  \{v_1,\dotsc,v_i\}]-\{\{v_1,\bar{n}_1\},\{v_i,\bar{n}_i\}\},P_i)$,
  where $\{\bar{n}_1\}=N_G(v_1) \setminus \{n_1\}$,
  $\{\bar{n}_i\}=N_G(v_i) \setminus \{n_i\}$, and 
  $P_i=\SB \{v_j,v_{j+1}\}\SM 1 \leq j < i\SE$. The tables
  $T_2,\dotsc,T_n$ are iteratively computed starting with $T_2$ as
  follows. For every $n_1 \in N_G(v_1)$ and $n_2 \in N_G(v_2)$, the table
  $T_2[n_1,n_2]$ is equal to $E_{n_1,n_2}$. Moreover, for every $i$
  with $3\leq i \leq n$, the table $T_i$ is obtained from the table
  $T_{i-1}$ as follows. For every $n_1 \in N_G(v_1)$ and $n_i \in
  N_G(v_i)$, the table $T_i[n_1,n_i]$ is equal to the union of
  $\SB E+E' \SM E \in T_{i-1}[n_1,n_{i-1}] \land E' \in
  E_{\bar{n}_{i-1},n_i} \SE$ and
  $\SB E+E' \SM E \in T_{i-1}[n_1,\bar{n}_{i-1}] \land E' \in
  E_{n_{i-1},n_i} \SE$, where $\{\bar{n}_1\}=N_G(v_1) \setminus \{n_1\}$,
  $\{\bar{n}_i\}=N_G(v_i) \setminus \{n_i\}$, and $\{n_{i-1},\bar{n}_{i-1}\}=N_G(v_{i-1})$.
  Finally, the set of all hypothetical solution vectors for the instance $(G[A\cup
  C],P_C)$ is equal to $\SB E \in T[n_1,n_n]\SM n_1 \in N_G(v_1)
  \land n_n \in N_G(v_n) \SE$, where $T[n_1,n_n]=\SB E+E' \SM E \in
  T_i[n_1,n_n] \land E' \in E_{\bar{n}_n,\bar{n}_1}\SE$ for every $n_1
  \in N_G(v_1)$, $n_n \in N_G(v_n)$, $\{\bar{n}_1\}=N_G(v_1)
  \setminus \{n_1\}$, and $\{\bar{n}_n\}=N_G(v_n) \setminus \{n_n\}$.

  The set $\LLL_H$ described in step (S2B) for a path $H=(v_1,\dotsc,v_n)$ of $G^T[B]$
  is computed as follows.
  As for the case of
  a cycle the algorithm starts by computing the table $T_n$, which for
  every $n_1 \in N_G(v_1)$ and $n_n \in N_G(v_n)$ contains all
  solution vectors for the instance $(G[A\cup V(P)]-E^{1,n},E(P))$,
  where $E^{1,n}$ contains the edge from $v_1$ to $N_G(v_1)\setminus
  \{n_1\}$ if $N_G(v_1)\setminus \{n_1\}\neq \emptyset$ and the edge
  from $v_n$ to $N_G(v_n)\setminus \{n_n\}$ if $N_G(v_n)\setminus
  \{n_n\}\neq \emptyset$. Let $T$ be equal to $T_n$, then the
  algorithm proceeds as follows. If $P'$ contains a terminal pair
  $\{v_1,a\}$ with $a \in A$, then for every $n_1 \in N_G(v_1)$ and
  every $n_n \in N_G(v_n)$, the algorithm updates $T[n_1,n_n]$ to be
  the set 
  $\SB E+E' \SM E \in T_n[n_1,n_n] \land E' \in E_{\bar{n}_1,a}\SE$,
  where $\{\bar{n}_1\}=N_G(v_1)\setminus \{n_1\}$. Similarly, if $P'$
  contains a terminal pair
  $\{v_n,a\}$ with $a \in A$, then for every $n_1 \in N_G(v_1)$ and
  every $n_n \in N_G(v_n)$, the algorithm updates $T[n_1,n_n]$ to be
  the set 
  $\SB E+E' \SM E \in T_n[n_1,n_n] \land E' \in E_{\bar{n}_n,a}\SE$,
  where $\{\bar{n}_n\}=N_G(v_n)\setminus \{n_n\}$. Finally,  
  the set of all solution vectors $\LLL_H$ for the instance $(G[A\cup
  V(H)],P')$ is obtained as the set $\SB E \in T[n_1,n_n]\SM n_1 \in N_G(v_1)
  \land n_n \in N_G(v_n)\SE$.

  This completes the description of the algorithm. To verify correctness, one can observe that each solution vector computed by the algorithm can be traced back to a specific choice of edges (a path) that connects each terminal pair in $P$, and since there are sufficient multi-edges in $G[A]$ to accommodate all the resulting paths, this guarantees the existence of a solution. On the other hand, if a solution exists then it surely has a solution vector, and moreover the algorithm will discover this solution vector by choosing, for each $\{a,b\} \in P$, the entry in $E_H$ which corresponds to the $a$-$b$ path used in the solution.
  
    Finally, we argue the running time bound. Note first that every set of
  solution vectors computed at any point in the algorithm contains at
  most $|P|^{\binom{k}{2}}$ elements. Moreover, as argued in (O1) the
  set $E_{u,v}$ for two distinct vertices $u$ and $v$ in $A$ can be
  computed in time $\bigoh((k+1)!)$ and contains at most $(k+1)!$
  elements. From this it follows that the time required to compute
  $\LLL_A$ in (S1) is at most $\bigoh(|P|^{\binom{k}{2}}(k+1)!|P'|)$.
  Similarly, the time required to compute $\LLL_C$ for a cycle $C$ in
  $G^T[B]$ in step (S2A) is at most
  $\bigoh(|P|^{\binom{k}{2}}(k+1)!|P_C|)$ and the time required to compute
  $\LLL_H$ for a path $H$ in
  $G^T[B]$ in step (S2B) is at most
  $\bigoh(|P|^{\binom{k}{2}}(k+1)!|P_H|)$. Hence the time required to
  compute $\LLL_A$ together with all the sets $\LLL_C$ and $\LLL_H$
  for every cycle $C$ and path $H$ of $G^T[B]$ is at most 
  $\bigoh(|P|^{\binom{k}{2}}(k+1)!|P|)$. Finally, combining these sets
  into $\LLL'$ does not incur an additional run-time overhead since
  $\LLL'$ can be computed iteratively as part of the computation of the
  sets $\LLL_A$, $\LLL_C$, and $\LLL_H$.
  }
\end{proof}

Notice that Lemma~\ref{lem:simplealg} does not provide a fixed-parameter algorithm for \textsc{Simple EDP}. Our second task for this section will be to rule out the existence of such algorithms (hence also ruling out the fixed-parameter tractability of \EDP\ parameterized by treecut width).

Before we proceed, we would like note that this outcome was highly surprising for the authors. Indeed, not only does this ``break'' the parallel between $\{$\VDP, treewidth$\}$ and $\{$\EDP, treecut width$\}$, but inspecting the dynamic programming algorithm for \EDP\ parameterized by treecut width presented in Section~\ref{sec:tcwalg} reveals that solving \textsc{Simple EDP} is the only step which requires more than ``FPT-time''. In particular, if \textsc{Simple EDP} were \FPT, then \EDP\ parameterized by treecut width would also be \FPT. This situation contrasts the vast majority of dynamic programming algorithms for parameters such as treewidth and clique-width~\cite{CourcelleMR00}, where the complexity bottleneck is usually tied to the size of the records used and not to the computation of the dynamic step.

\sv{Our lower-bound result is based on a parameterized
reduction from the following problem, the \W{1}-hardness of which follows from recent work of the authors~\cite{GanianOrdyniakSridharan17,GanianKO18} ($\star$):
\pbDefP{\textsc{Multidimensional Subset Sum (MSS)}}{An
  integer $k$, a set $S=\{s_1,\dotsc,s_n\}$ of item-vectors with $s_i
  \in \Nat^{k}$ for every $i$ with $1\leq i \leq n$, a target vector
  $t \in \Nat^k$, and an integer $\ell$.}{$k$}{Is there a subset $S'
  \subseteq S$ with $|S'|\geq \ell$ such that $\sum_{s \in S'}s\leq t$?}

}
\lv{
Our lower-bound result is based on a parameterized
reduction from the following problem:
\pbDefP{\textsc{Multidimensional Subset Sum (MSS)}}{An
  integer $k$, a set $S=\{s_1,\dotsc,s_n\}$ of item-vectors with $s_i
  \in \Nat^{k}$ for every $i$ with $1\leq i \leq n$, a target vector
  $t \in \Nat^k$, and an integer $\ell$.}{$k$}{Is there a subset $S'
  \subseteq S$ with $|S'|\geq \ell$ such that $\sum_{s \in S'}s\leq t$?}
The \W{1}\hy hardness of MSS can be obtained by a trivial reduction from the following problem, which was recently shown to be \W{1}\hy hard by Ganian, Ordyniak and Ramanujan~\cite{GanianOrdyniakSridharan17}:
\pbDefP{\textsc{Multidimensional Relaxed Subset Sum (MRSS)}}{An
  integer $k$, a set $S=\{s_1,\dotsc,s_n\}$ of item-vectors with $s_i
  \in \Nat^{k}$ for every $i$ with $1\leq i \leq n$, a target vector
  $t \in \Nat^k$, and an integer $\ell$.}{$k$}{Is there a subset $S'
  \subseteq S$ with $|S'|\leq \ell$ such that $\sum_{s \in S'}s\geq t$?}
Indeed, given an instance $(k,S,t,\ell)$ of MRSS, it is straightforward to
verify that $(k,S,(\sum_{s \in S}s)-t,|S|-\ell)$ is an
equivalent instance of MSS; since the reduction preserves the
parameter, this shows that MSS is also \W{1}\hy hard. 
}

\lv{\begin{lemma}}
\sv{\begin{lemma}[$\star$]}
\label{lem:simplehard}
\textsc{Simple EDP} is \W{1}\hy hard.
\end{lemma}
\lv{\begin{proof}}
\sv{\begin{proof}[Sketch of Proof.]}
  We provide a parameterized reduction from MSS. Namely, given an
  instance $(k,S,t,\ell)$ of MSS, we will construct an equivalent instance $(G,P)$
  with partition $A$ and $B$ and $|A|=k+3$ of \textsc{Simple EDP}. For
  convenience and w.l.o.g.\ we will assume that all entries of the
  vectors in $S$ as well as all entries of the target vector $t$ are
  divisible by two; furthermore, we will describe the constructed instance of \textsc{Simple EDP} with multi-edges between vertices in $A$ (note that these can be replaced by degree-2 vertices in $B$, similarly as in Lemma~\ref{lem:simplealg}).
  
  \lv{The graph $G[A]$ has vertices $a$, $b$, $d$, and $d_1,\dotsc,d_k$
  and the following multi-edges:
  \begin{itemize}
  \item $|S|-\ell$ edges between $a$ and $b$,
  \item for every $i$ with $1 \leq i \leq k$, $t[i]$ edges between $d$
    and $d_i$.
  \end{itemize}}\sv{The graph $G[A]$ has vertices $a$, $b$, $d$, and $d_1,\dotsc,d_k$
  and the following multi-edges: (1) $|S|-\ell$ edges between $a$ and $b$, and
  (2) for every $i$ with $1 \leq i \leq k$, $t[i]$ edges between $d$ and $d_i$.}
  Moreover, for every $s \in S$ we construct a gadget $G(s)$
  consisting of:
  \begin{itemize}
  \item the vertices $v^s,v^s_1,u^s_1,\dotsc,v^s_{\bar{s}},u^s_{\bar{s}}$ with
    $\bar{s}=\sum_{i=1}^k s[i]$,
  \item two edges $\{v^s,a\}$ and $\{v^s,d\}$, 
  \item for every $i$ with $1 \leq i \leq \bar{s}$, two edges
    $\{v_i^s,b\}$ and $\{u_i^s,b\}$,
  \item for every $i$ with $1\leq i \leq \bar{s}$ and $i$ even, two
    edges $\{v^s_i,d\}$ and $\{u^s_i,d\}$,
  \item for every $j$ with $1 \leq j \leq k$ and every $i$ with
    $\sum_{l=1}^{j-1}s[l] < i \leq \sum_{l=1}^{j}s[l]$ and $i$ odd,
    two edges $\{v^s_i,d_j\}$ and $\{u^s_i,d_j\}$,
  \item the terminal pair $\{v^s,v^s_1\}$,
  \item for every $i$ with $1 \leq i \leq \bar{s}$, a
    terminal pair $\{v^s_i,u^s_i\}$,
  \item for every $i$ with $1 \leq i < \bar{s}$, a
    terminal pair $\{u^s_i,v^s_{i+1}\}$,
  \end{itemize}

  \begin{figure}[ht]
    \begin{center}
    \vspace{-0.3cm}
      \begin{tikzpicture}[node distance=1cm]
        \tikzstyle{every node}=[]
        \tikzstyle{gn}=[draw, shape=circle, minimum size=3pt,inner
        sep=2pt, fill=white]
        \tikzstyle{ge}=[draw, line width=2pt]
        \tikzstyle{geg}=[draw, line width=0pt]

        \tikzstyle{gec1}=[color=green]
        \tikzstyle{gec2}=[]
        
        \tikzstyle{gegc1}=[color=green]
        \tikzstyle{gegc2}=[]

        \draw
        node[gn, label=above:$a$] (a) {}
        node[gn, node distance=2.5cm, below of=a, label=below:$b$] (b) {}

        node[node distance=1.5cm, right of=a] (vsh) {}
        node[gn, node distance=1.5cm, above of=vsh,label=above:$v^s$] (vs) {}
        node[gn, below of=vs, label=above:$v^s_1$] (vs1) {}
        node[gn, below of=vs1, label=above:$u^s_1$] (us1) {}
        
        node[gn, below of=us1, label=above:$v^s_2$] (vs2) {}
        node[gn, below of=vs2, label=above:$u^s_2$] (us2) {}

        node[gn, below of=us2, label=above:$v^s_3$] (vs3) {}
        node[gn, below of=vs3, label=above:$u^s_3$] (us3) {}

        node[gn, below of=us3, label=above:$v^s_4$] (vs4) {}
        node[gn, below of=vs4, label=above:$u^s_4$] (us4) {}

        node[node distance=1.5cm, node distance=0.5cm, below of=vs1] (d1h){}
        node[gn, node distance=1.5cm, right of=d1h, label=above:$d_1$] (d1) {}
        node[node distance=3cm, below of=d1] (dd) {}
        node[gn, node distance=2.5cm, below of=dd, label=below:$d_2$] (d2) {}
        node[gn, node distance=1.5cm, right of=dd, label=right:$d$] (d) {}
        ;
        \draw[ge,gec1]
        (a) -- (b) node[midway, left=5pt,color=black] {$|S|-\ell$}
        ;
        \draw[ge,gec2]
        (d) -- (d1) node[midway, right,color=black] {$t[1]$}
        (d) -- (d2) node[midway, right,color=black] {$t[2]$}
        ;
        \draw
        (vs) edge[geg, gegc2, bend left=45] (d)
        (vs) edge[geg, gegc1] (a)
        
        (vs1) edge[geg,gegc1, gegc2] (b)
        (us1) edge[geg,gegc1, gegc2] (b)
        (vs2) edge[geg,gegc1, gegc2] (b)
        (us2) edge[geg,gegc1, gegc2] (b)
        (vs3) edge[geg,gegc1, gegc2] (b)
        (us3) edge[geg,gegc1, gegc2] (b)
        (vs4) edge[geg,gegc1, gegc2] (b)
        (us4) edge[geg] (b)

        (vs2) edge[geg,gegc1, gegc2] (d)
        (us2) edge[geg,gegc1, gegc2] (d)
        (vs4) edge[geg,gegc1, gegc2] (d)
        (us4) edge[geg,gegc1, gegc2] (d)

        (vs1) edge[geg,gegc1, gegc2] (d1)
        (us1) edge[geg,gegc1, gegc2] (d1)
        (vs3) edge[geg,gegc1, gegc2] (d2)
        (us3) edge[geg,gegc1, gegc2] (d2)
        
        ;
        
      \end{tikzpicture}
      \begin{tikzpicture}[node distance=1cm]
        \tikzstyle{every node}=[]
        \tikzstyle{gn}=[draw, shape=circle, minimum size=3pt,inner
        sep=2pt, fill=white]
        \tikzstyle{ge}=[draw, line width=2pt]
        \tikzstyle{geg}=[draw, line width=0pt]

        \tikzstyle{gec1}=[]
        \tikzstyle{gec2}=[color=green]
        
        \tikzstyle{gegc1}=[]
        \tikzstyle{gegc2}=[color=green]

        \draw
        node[gn, label=above:$a$] (a) {}
        node[gn, node distance=2.5cm, below of=a, label=below:$b$] (b) {}

        node[node distance=1.5cm, right of=a] (vsh) {}
        node[gn, node distance=1.5cm, above of=vsh,label=above:$v^s$] (vs) {}
        node[gn, below of=vs, label=above:$v^s_1$] (vs1) {}
        node[gn, below of=vs1, label=above:$u^s_1$] (us1) {}
        
        node[gn, below of=us1, label=above:$v^s_2$] (vs2) {}
        node[gn, below of=vs2, label=above:$u^s_2$] (us2) {}

        node[gn, below of=us2, label=above:$v^s_3$] (vs3) {}
        node[gn, below of=vs3, label=above:$u^s_3$] (us3) {}

        node[gn, below of=us3, label=above:$v^s_4$] (vs4) {}
        node[gn, below of=vs4, label=above:$u^s_4$] (us4) {}

        node[node distance=1.5cm, node distance=0.5cm, below of=vs1] (d1h){}
        node[gn, node distance=1.5cm, right of=d1h, label=above:$d_1$] (d1) {}
        node[node distance=3cm, below of=d1] (dd) {}
        node[gn, node distance=2.5cm, below of=dd, label=below:$d_2$] (d2) {}
        node[gn, node distance=1.5cm, right of=dd, label=right:$d$] (d) {}
        ;
        \draw[ge,gec1]
        (a) -- (b) node[midway, left=5pt,color=black] {$|S|-\ell$}
        ;
        \draw[ge,gec2]
        (d) -- (d1) node[midway, right,color=black] {$t[1]$}
        (d) -- (d2) node[midway, right,color=black] {$t[2]$}
        ;
        \draw
        (vs) edge[geg, gegc2, bend left=45] (d)
        (vs) edge[geg, gegc1] (a)
        
        (vs1) edge[geg,gegc1, gegc2] (b)
        (us1) edge[geg,gegc1, gegc2] (b)
        (vs2) edge[geg,gegc1, gegc2] (b)
        (us2) edge[geg,gegc1, gegc2] (b)
        (vs3) edge[geg,gegc1, gegc2] (b)
        (us3) edge[geg,gegc1, gegc2] (b)
        (vs4) edge[geg,gegc1, gegc2] (b)
        (us4) edge[geg,gegc1, gegc2] (b)

        (vs2) edge[geg,gegc1, gegc2] (d)
        (us2) edge[geg,gegc1, gegc2] (d)
        (vs4) edge[geg,gegc1, gegc2] (d)
        (us4) edge[geg] (d)

        (vs1) edge[geg,gegc1, gegc2] (d1)
        (us1) edge[geg,gegc1, gegc2] (d1)
        (vs3) edge[geg,gegc1, gegc2] (d2)
        (us3) edge[geg,gegc1, gegc2] (d2)
        
        ;
        
      \end{tikzpicture}
    \end{center}
    \vspace{-0.4cm}
\sv{\caption{An illustration of the graph $G[A]$ together with the gadget $G(s)$ for $k=2$, $s[1]=2$, and $s[2]=2$. Bold edges indicate multi-edges with multiplicities given as a label. The left side depicts (in green) a solution which uses the edge $\{a,b\}$ (corresponding to not ``picking'' $s$), while the right side depicts a solution which uses edges incident to $d$ (corresponding to ``picking'' $s$).
}}
    \lv{\caption{An illustration of the graph $G[A]$ together with the gadget $G(s)$ for $k=2$, $s[1]=2$, and $s[2]=2$. Bold edges indicate multi-edges with multiplicities given as an edge label. The left side illustrates configuration (C1) and the right side illustrates configuration (C2) as defined in Claim~\ref{clm:simple-edp-hard}; here the green edges indicate the edges used by a solution that uses the corresponding configuration to connect the terminal pairs of $G(s)$.}}
    \vspace{-0.3cm}
    \label{fig:simple-edp-hard}
  \end{figure}

  Then $G$ consists of the graph $G[A]$ together with the vertices and
  edges of the gadget $G(s)$ for every $s \in S$; note that $B$ is the
  union of the vertices of the gadgets $G(s)$ for every $s \in
  S$. Moreover, $P$ consists of all terminal pairs of the gadgets
  $G(s)$ for every $s \in S$.
  This completes the construction of the instance $(G,P)$; an illustration is
  provided in Figure~\ref{fig:simple-edp-hard}. 
  \sv{
   Intuitively, each gadget $G(s)$ forces us to make a choice: either connect all terminal pairs in the gadget by using an $\{a,b\}$ edge once (in this case no other edges outside of $G(s)$ are necessary), or avoid the edge $\{a,b\}$ (in which case we need to use precisely $s[i]$-many edges between each $d_i$ and $d$). Since at least $\ell$ gadgets need to be routed without using an $\{a,b\}$ edge, it follows that   $(k,S,t,\ell)$ has a solution if and only if the resulting instance of \EDP\ has a solution. 
    }  
  \lv{
  It remains to show that
  the instance $(k,S,t,\ell)$ of MSS has a solution if and only if so
  does the instance $(G,P)$ of EDP. 

  We start by showing that there are only two ways to connect all
  terminal pairs of the gadget $G(s)$ for every $s \in
  S$. Figure~\ref{fig:simple-edp-hard} illustrates the edges
  used by the two configurations.
  \begin{CLM}\label{clm:simple-edp-hard}
    Let $\SSS$ be a solution for $(G,P)$, and $s \in S$. Then either:
    \begin{itemize}
    \item[(C1)] The terminal pair $\{v^s,v^s_1\}$ is connected by the path
      $(v^s,a,b,v^s_1)$ and:
      \begin{itemize}
      \item for every $i$ with $1\leq i < \bar{s}$, the
        terminal pair $\{u^s_i,v^s_{i+1}\}$ is connected by the path
        $(u^s_i,b,v^s_{i+1})$,
      \item for every $i$ with $1\leq i \leq \bar{s}$ and $i$ even,
        the terminal pair $\{v^s_i,u^s_i\}$ is connected by the path
        $(v^s_i,d,u^s_i)$, and 
      \item for every $i$ with $1\leq i \leq \bar{s}$ and $i$ odd,
        the terminal pair $\{v^s_i,u^s_i\}$ is connected by the path
        $(v^s_i,d_j,u^s_i)$, where $j$ is such that
        $\sum_{l=1}^{j-1}s_l <i\leq\sum_{l=1}^{j}$.
      \end{itemize}
    \item[(C2)] The terminal pair $\{v^s,v^s_1\}$ is connected by the path
      $(v^s,d,d_j,v^s_1)$, where $j$ is the minimum integer such that
      $s[j]\neq 0$ and:
      \begin{itemize}
      \item for every $i$ with $1\leq i \leq \bar{s}$, the
        terminal pair $\{v^s_i,u^s_{i}\}$ is connected by the path
        $(v^s_i,b,u^s_{i})$,
      \item for every $i$ with $1\leq i < \bar{s}$ and $i$ is odd, the
        terminal pair $\{u^s_i,v^s_{i+1}\}$ is connected by the path
        $(u^s_i,d_j,d,v^s_{i+1})$, where $j$ is such that
        $\sum_{l=1}^{j-1}s[l] < i \leq \sum_{l=1}^js[l]$,
      \item for every $i$ with $1\leq i < \bar{s}$ and $i$ is even, the
        terminal pair $\{u^s_i,v^s_{i+1}\}$ is connected by the path
        $(u^s_i,d,d_j,v^s_{i+1})$, where $j$ is such that
        $\sum_{l=1}^{j-1}s[l] < i \leq \sum_{l=1}^js[l]$.
      \end{itemize}
    \end{itemize}
  \end{CLM}
  \begin{proof}
    Let $\SSS$ be a solution for $(G,P)$ and $s \in G(s)$. Then $\SSS$ has
    to connect the terminal pair $\{v^s,v^s_1\}$ either by the path
    $(v^s,a,b,v^s_1)$ or by the path $(v^s,d,d_j,v^s_1)$.

    In the
    former case, the only way to connect the terminal pair
    $\{v^s_1,u^s_1\}$ is the path $(v^s_1,d_j,u^s_1)$, where $j$ is
    such that $\sum_{l=1}^{j-1}s_l <1\leq\sum_{l=1}^{j}$. But then the
    terminal pair $\{u^s_1,v^s_2\}$ can only be connected by the path
    $(u^s_1,b,v^s_2)$ and in turn the terminal pair $\{v^s_2,u^s_2\}$
    can only be connected by the path $(v^s_2,d,u^s_2)$. Since this
    pattern continues in this manner, this concludes the
    argument for the first case.

    In the later case, the only way to connect the terminal pair
    $\{v^s_1,u^s_1\}$ is the path $(v^s_1,b,u^s_1)$. But then the
    terminal pair $\{u^s_1,v^s_2\}$ can only be connected by the path
    $(u^s_1,d_j,d,v^s_2)$, where $j$ is
    such that $\sum_{l=1}^{j-1}s_l <1\leq\sum_{l=1}^{j}$, and in turn the terminal pair $\{v^s_2,u^s_2\}$
    can only be connected by the path $(v^s_2,b,u^s_2)$. Finally, the
    terminal pair $\{u^s_2,v^s_3\}$ can then only be connected by the path
    $(u^s_2,d,d_j,v^s_3)$, where $j$ is
    such that $\sum_{l=1}^{j-1}s_l <1\leq\sum_{l=1}^{j}$. Since this
    pattern continues in this manner, this concludes the
    argument for the second case.
  \end{proof}
  Let $\SSS$ be a solution for $(G,P)$ and $s \in S$. It follows from
  Claim~\ref{clm:simple-edp-hard} that if $\SSS$ connects the terminal
  pairs of $G(s)$ according to (C1), then the only edge used
  from $G[A]$ is the edge $\{a,b\}$. On the other hand, if $\SSS$ connects the
  terminal pairs in $G(s)$ according to (C2), then $\SSS$ uses 
  $s[i]$ edges between $d$ and $d_j$ for every $i$ with $1\leq i \leq k$.

  Towards showing the forward direction, let $S' \subseteq S$ be a solution for $(k,S,t,\ell)$. 
  W.l.o.g. we can assume that $|S'|=\ell$. We claim that the set of
  edges disjoint paths $\SSS$, which if $s \in S'$ connects all
  terminal pairs in $G(s)$ according to (C2) and if $s \in S\setminus
  S'$ connects all terminal pairs in $G(s)$ according to (C1) is a
  solution for $(G,P)$. This holds because there are $|S|-\ell$ edges
  between $a$ and $b$, which are sufficient for the elements in
  $S\setminus S'$ to be connected according to (C1). Moreover, because
  $\sum_{s\in S'}s\leq t$, the $t[i]$ edges between $d$ and $d_i$ for
  every $i$ with $1 \leq i \leq k$, suffices for the elements in $S'$
  to be connected according to (C2).

  For the reverse direction, let
  $\SSS$ be a solution for $(G,P)$.
  We claim that the subset $S'$ of $S$ containing all $s \in S$ such
  that $\SSS$ connects all terminal pairs in $G(s)$ according to C2 is
  a solution for $(k,S,t,\ell)$. Because there are at most $|S|-\ell$
  edges between $a$ and $b$ in $G[A]$, we obtain that $|S'|\geq \ell$.
  Moreover, because there are at most $t[i]$ edges between $d$ and
  $d_i$ in $G[A]$, it follows that $\sum_{s\in S'}s\leq t$.
  Consequently, $S'$ is a solution for $(k,S,t,\ell)$.}
\end{proof}

\section{An algorithm for \EDP\ for graphs of bounded treecut width}
\label{sec:tcwalg}
The goal of this section is to provide an \XP\ algorithm for
\textsc{EDP} parameterized by treecut-width. The core of the algorithm
is a dynamic programming procedure which runs on a nice treecut
decomposition $ (T,\mathcal{X}) $ of the input graph $G$.

\lv{\subsection{Overview}}
\sv{\smallskip
\noindent \textbf{Overview.}\quad}
Our first aim is to define the data table the algorithm is going to dynamically compute for individual nodes of the treecut decomposition; to this end, we introduce two additional notions. For a node $t$, we say that $Y_t$ (or $G_t$) contains an \emph{unmatched} terminal $s$ if $\{s,t\}\in P$, $s\in Y_t$ and $t\not \in Y_t$; let $U_t$ be the multiset containing all unmatched terminals $Y_t$ (one entry in $U_t$ per tuple in $P$ which contains an unmatched terminal). For a subgraph $H$ of $G$, let $P_H\subseteq P$ denote the subset of terminal pairs whose both endpoints lie in $H$.

Let a \emph{record} for node $t$ be a tuple $(\delta, I, F, L)$ where:
\begin{itemize}
\item $\delta$ is a partitioning of $E_t$ into \emph{internal} ($I'$), \emph{leaving} ($L'$), \emph{foreign} ($F'$) and \emph{unused} $(U')$;
\item $I$ is a set of subsets of size $2$, which forms a perfect matching between the edges in $I'$;
\item $F$ is a set of subsets of size $2$, which forms a perfect matching between the edges in $F'$;
\item $L$ is a perfect matching between $U_t$ and the edges in $L'$.
\end{itemize}

Intuitively, a record captures all the information we need about one possible interaction between a solution to \textsc{EDP} and the edges in $E_t$. In particular, unmatched terminals need to cross between $Y_t$ and $G_t$ using an edge in $E_t$ and $L$ captures the first edge used by a path from an unmatched terminal in the solution, while $I$ and $F$ capture information about paths which intersect with $E_t$ but whose terminals both lie in $Y_t$ and $V(G_t)\setminus Y_t$, respectively. We formalize this intuition below through the notion of a \emph{valid record}.

\begin{definition}
\label{def:valid}
A record $\lambda=(\delta, I, F, L)$ is \emph{valid} for $t$ if $(G^\lambda,P^\lambda)$ is a \textbf{YES}-instance of \textsc{EDP}, where $(G^\lambda,P^\lambda)$ is constructed from $(G_t, P_{G_t})$ as follows:
\begin{enumerate}
\item For each $\{\{a,b\},\{c,d\}\}\in I$ where $a,c\in Y_t$, add a new vertex into $G_t$ and connect it to $a$ and $c$ by edges (note that if $a=c$ then this simply creates a new leaf and hence this operation can be ignored).
\item For each $\{s,\{a,b\}\}\in L$ where $a\in Y_t$, add a new tuple $\{s,t'\}$ into $P_{G_t}$ and a new leaf $t'$ into $G_t$ adjacent to $a$.
\item For each $\{\{a,b\},\{c,d\}\}\in F$ where $a,c\in Y_t$, add two new leaves $b', d'$ into $G_t$, make them adjacent to $a$ and $c$ respectively, and add $\{b',d'\}$ into $P_{G_t}$.
\end{enumerate}
\end{definition}

We are now ready to define our data tables: for a node $t\in V(T)$, let $D(t)$ be the set of all valid records for $t$. We now make two observations. First, for any node $t$ in a nice treecut decomposition of width $k$, it holds that there exist at most $4^k\cdot k!$ distinct records and hence $|D(t)|\leq 4^k\cdot k!$; indeed, there are $4^k$ possible choices for $\delta$, and for each such choice and each edge $e$ in $E_t$ one has at most $k$ options of what to match with $e$. Second, if $r$ is the root of $T$, then either $D(r)=\emptyset$ or $D(r)=\{(\emptyset, \emptyset, \emptyset, \emptyset)\}$; furthermore, $(G,P)$ is a \textbf{YES}-instance if and only if the latter holds. Hence it suffices to compute $D(r)$ in order to solve \textsc{EDP}. 

The next lemma shows that $D(t)$ can be computed efficiently for all leaves of $t$.

\lv{\begin{lemma}}
\sv{\begin{lemma}[$\star$]}
\label{lem:tcwleaves}
There is an algorithm which takes as input $(G,P)$, a width-$k$ treecut decomposition $(T,\mathcal{X})$ of $G$ and a leaf $t\in V(T)$, runs in time $k^{\bigoh(k^2)}$, and outputs $D(t)$.
\end{lemma}

\lv{
\begin{proof}
We proceed as follows. For each record $\lambda$ for $t$, we construct the instance $(G^\lambda,P^\lambda)$ as per Definition~\ref{def:valid} and check whether $(G^\lambda,P^\lambda)$ is a \textbf{YES}-instance of \textsc{EDP}. Since $V(G^\lambda)\leq 2k$, a simple brute-force algorithm will suffice here. For instance, one can enumerate all partitions of the at most $4k^2$ edges in $G^\lambda$, and for each such partition one can check whether this represents a set of edge disjoint paths which forms a solution to $(G^\lambda,P^\lambda)$.
If $(G^\lambda,P^\lambda)$ is a \textbf{YES}-instance of \textsc{EDP} then we add $\lambda$ into $D(t)$, and otherwise we do not.

The number of partitions of a set of size $4k^2$ is upper-bounded by $k^{\bigoh(k^2)}$~\cite{BerendTassa10}, and $|D(t)|\leq 4^k\cdot k!$. Hence the runtime of the whole algorithm described above is dominated by $k^{\bigoh(k^2)}$.
\end{proof}
}

At this point, all that is left to obtain a dynamic leaves-to-root algorithm which solves \textsc{EDP} is the dynamic step, i.e., computing the data table for a node $t\in V(t)$ from the data tables of its children. Unfortunately, that is where all the difficulty of the problem lies, and our first step towards handling this task will be the introduction of two additional notions related to records. The first is \emph{correspondence}, which allows us to associate each solution to $(G,P)$ with a specific record for $t$; on an intuitive level, a solution corresponds to a particular record if that record precisely captures the ``behavior'' of that solution on $E_t$. Correspondence will, among others, be used to later argue the correctness of our algorithm.

\begin{definition}
\label{def:correspondence}
A solution $\mathcal{S}$ to $(G,P)$ \emph{corresponds} to a record $\lambda=(\delta, I, F, L)$ for $t$ if the conditions \textbf{1.}-\textbf{4.} stated below hold for every $a$-$b$ path $S\in \mathcal{S}$ such that $S\cap E_t\neq \emptyset$. We let $s=|S\cap E_t|$ and we denote individual edges in $S\cap E_t$ by $e_1,e_2,\dots e_{s}$, ordered from the edge nearest to $a$ along $S$.
\begin{enumerate}
\item If $a,b\not \in Y_t$, then for each odd $i\in [s]$, $F$ contains $(e_i,e_{i+1})$.
\item If $a,b \in Y_t$, then for each odd $i\in [s]$, $I$ contains $(e_i,e_{i+1})$.
\item If $\{a,b\} \cap Y_t=\{a\}$, then $L$ contains $(a,e_1)$, and for each even $i\in [s]$ $F$ contains $(e_{i},e_{i+1})$.
\item There are no elements in $I,F,L$ other than those specified above.
\end{enumerate}
\end{definition}

Note that ``restricting'' the solution $\mathcal{S}$ to the instance $(G^\lambda, P^\lambda)$ used in Definition~\ref{def:valid} yields also a solution to $(G^\lambda, P^\lambda)$; in particular, for each path $S\in \mathcal{S}$ that intersects $E_t$, one replaces the path segments of $S$ in $G\setminus Y_t$ by the newly created vertices to obtain a solution to $(G^\lambda, P^\lambda)$. 
Consequently, if $\mathcal{S}$ corresponds to $\lambda$ then $\lambda$ must be valid (however, it is clearly not true that every valid record has a solution to the whole instance that corresponds to it). Moreover, since Definition~\ref{def:correspondence} is constructive and deterministic, for each solution $\mathcal{S}$ and node $t$ there exists precisely one corresponding valid record $\lambda$. 

The second notion that we will need is that of \emph{simplification}. This is an operation which takes a valid record $\lambda$ for a node $t$ and replaces $G_t$ by a ``small representative'' so that the resulting graph retains the existence of a solution corresponding to $\lambda$. Simplification can also be seen as being complementary to the construction of $(G^\lambda,P^\lambda)$ used in Definition~\ref{def:valid} (instead of modeling the implications of a record on $G_t$, we model its implications on $G-Y_t$), and will later form an integral part of our procedure for computing valid records for nodes.

\begin{definition}
\label{def:simplification}
\sv{
The \emph{simplification} of a node $t$ in accordance with $\lambda=(\delta, I, F, L)$ is an operation which transforms the instance $(G,P)$ into a new instance $(G',P')$ obtained from $(G-Y_t,P_{G-Y_t})$ as follows. (1) For each $\{s,\{a,b\}\}\in L$ where $(s,t)\in P$ and $b\not \in Y_t$, we add $(s,t)$ to $P'$ and create a vertex $s$ adjacent to $b$. (2) For each $\{\{a,b\},\{c,d\}\}\in I$ where $a,c\in Y_t$ and $a\neq c$, we add vertices $a$ and $c$ into $G'$ and make them adjacent to $b$ and $d$ respectively, and add $(a,c)$ into $P'$. (3) For each $\{\{a,b\},\{c,d\}\}\in F$ where $a,c\in Y_t$ and $b\neq d$, we create a vertex $e$ and set $N(e)=\{b,d\}$.
}
\lv{
The \emph{simplification} of a node $t$ in accordance with $\lambda=(\delta, I, F, L)$ is an operation which transforms the instance $(G,P)$ into a new instance $(G',P')$ obtained from $(G-Y_t,P_{G-Y_t})$ as follows:
\begin{itemize}
\item For each $\{s,\{a,b\}\}\in L$ where $(s,t)\in P$ and $b\not \in Y_t$, add $(s,t)$ to $P'$ and create a vertex $s$ adjacent to $b$.
\item For each $\{\{a,b\},\{c,d\}\}\in I$ where $a,c\in Y_t$ and $a\neq c$, add vertices $a$ and $c$ into $G'$ and make them adjacent to $b$ and $d$ respectively, and add $(a,c)$ into $P'$.
\item For each $\{\{a,b\},\{c,d\}\}\in F$ where $a,c\in Y_t$ and $b\neq d$, create a vertex $e$ and set $N(e)=\{b,d\}$.
\end{itemize}
}
\end{definition}

\lv{
With regards to simplification, observe that every vertex added to $G-Y_t$ has degree at most $2$ and that simplification can never increase the degree of vertices in $G-Y_t$.
}

\lv{\begin{observation}}
\sv{\begin{observation}[$\star$]}
\label{obs:simplification}
If there exists a solution to $(G,P)$ which \emph{corresponds} to a record $\lambda=(\delta, I, F, L)$ for $t$, and if $(G',P')$ is the result of simplification of $t$ in accordance with $\lambda$, then $(G',P')$ admits a solution. On the other hand, if $(G',P')$ is the result of simplification of $t$ in accordance with a record $\lambda$ and if $(G',P')$ admits a solution, then $(G,P)$ also admits a solution.
\end{observation}

\lv{
\begin{proof}
For the forward direction, consider a solution $S$ to $(G,P)$ which corresponds to $\lambda=(\delta, I, F, L)$. By comparing Definition~\ref{def:correspondence} with Definition~\ref{def:simplification}, we observe the following:
\begin{enumerate}
\item for each $s$-$t$ path $P\in S$ such that $s,t\not \in Y_t$ and $P\cap E_t\neq \emptyset$, it holds that each path segment of $P$ in $Y_t$ begins and ends with a pair of edges in $F$ and in particular is replaced by a single vertex in $(G',P')$;
\item for each $s$-$t$ path $P\in S$ such that $s,t\in Y_t$ and $P\cap E_t\neq \emptyset$, it holds that each path segment of $P$ outside of $Y_t$ begins and ends with a pair of edges in $I$ and in particular is replaced by a pair of new terminals in $(G',P')$;
\item for each $s$-$t$ path $P\in S$ such that $\{s,t\}\cap Y_t=\{s\}$, it holds that the path segment of $P$ in $Y_t$ containing $s$ ends with an edge in $L$ and is replaced by a new terminal in $(G',P')$, and all other path segments of $P$ in $Y_t$ begin and end with a pair of edges in $F$ and are hence replaced by single vertices in $(G',P')$.
\end{enumerate}
From the above, we observe that $S$ can be transformed into a solution $S'$ for $(G',P')$. The backward direction then follows by reversing the above observations; in particular, given a solution $S'$ for $(G',P')$, we use the fact that $\lambda$ is valid to expand $S'$ into a full solution $S$ to $(G,P)$.
\end{proof}
}

\sv{\smallskip
\noindent \textbf{The Dynamic Step.}\quad The following crucial lemma represents the tool that allows us to deal with the dynamic step of our leaf-to-root computation along the treecut decomposition.

\begin{lemma}[$\star$]
\label{lem:tcwdyn}
There is an algorithm which takes as input $(G,P)$ along with a width-$k$ treecut decomposition $(T,\mathcal{X})$ of $G$ and a non-leaf node $t\in V(T)$ and $D(t')$ for every child $t'$ of $t$, runs in time $(k|P|)^{\bigoh(k^2)}$, and outputs $D(t)$.
\end{lemma}

\begin{proof}[Sketch of Proof.]
We begin by looping through all of the at most $4^k\cdot k!$ distinct records for $t$; for each such record $\lambda$, our task is to decide whether it is valid, i.e., whether $(G^\lambda,P^\lambda)$ is a \textbf{YES}-instance. On an intuitive level, our aim will now be to use branching and simplification in order to reduce the question of checking whether $\lambda$ is valid to an instance of \textsc{Simple EDP}.

In our first layer of branching, we will select a record from the data tables of each node in $A_t$. Formally, we say that a \emph{record-set} is a mapping $\tau:t'\in A_t\mapsto \lambda_{t'}\in D(t')$. Note that the number of record-sets is upper-bounded by $(4^k\cdot k!)^{3(2k+1)}$, and we will loop over all possible record-sets.

Next, for each record-set $\tau$, we will apply simplification to each node $t'\in A_t$ in accordance with $\tau(t')$, and recall that each vertex $v$ created by this sequence of simplifications has degree at most $2$. We then apply a reduction rule to ensure that each such vertex is only adjacent to $(V(G)\setminus Y_t)\cup X_t$. At this point, every vertex contained in a bag $X_{t'}$ for $t'\in A_t$ has degree at most $2$ and is only adjacent to $X_t\cup(V(G)\setminus Y_t)$. Now, we apply one additional reduction rule which allows us to replace every thin node by vertices of degree at most $2$ adjacent to $X_t$ (there are only a few possible kinds of records that thin nodes can have, and we use simple replacement rules for each individual case). 

At this point, every vertex in $V(G^\lambda)\setminus X_t$ is of degree at most $2$ and only adjacent to $X_t$, and so $(G^\lambda,P^\lambda)$ is an instance of \textsc{Simple EDP}. All that is left is to invoke Lemma~\ref{lem:simplealg}; if it is a \textbf{YES}-instance then we add $\lambda$ to $D(t)$, and otherwise we do not.

We conclude the proof by arguing correctness. Assume $\lambda$ is a valid record. By Definition~\ref{def:valid}, this implies that $(G^\lambda,P^\lambda)$ admits a solution $S$. For each child $t'\in A_t$, $S$ corresponds to some record $\lambda^S_{t'}$ for $t$; consider now the branch in our algorithm which sets $\tau(t')=\lambda^S_{t'}$. Then by Observation~\ref{obs:simplification} it follows that each simplification carried out by the algorithm preserves the existence of a solution to $(G^\lambda,P^\lambda)$. Hence the instance of \textsc{Simple EDP} we obtain at the end of this branch must also be a \textbf{YES}-instance.
\end{proof}
}

\lv{\subsection{The Dynamic Step}
Let us begin by formalizing our aim for this subsection.

\begin{lemma}
\label{lem:tcwdyn}
There is an algorithm which takes as input $(G,P)$ along with a width-$k$ treecut decomposition $(T,\mathcal{X})$ of $G$ and a non-leaf node $t\in V(T)$ and $D(t')$ for every child $t'$ of $t$, runs in time $(k|P|)^{\bigoh(k^2)}$, and outputs $D(t)$.
\end{lemma}

Finally, we introduce two simple reduction rules which will later help us reduce our problem to \textsc{Simple EDP}. The first ensures that two vertices of degree at most $2$ are not adjacent to each other.

\begin{redrule}
\label{red:degtwored}
Let $(G,P)$ be an instance of \textsc{EDP} containing an edge $\{a,b\}$ between two vertices of degree at most $2$.
\begin{enumerate}
\item If $a$ is not a terminal, then contract $\{a,b\}$ and replace all occurrences of $b$ in $P$ by the new vertex;
\item If $\{a,b\}\in P$, then contract $\{a,b\}$ and replace all occurrences of $a$ and $b$ in $P$ by the new vertex;
\item If $\{a,b\}\not \in P$ and each of $a$ and $b$ occurs in precisely one element of $P$, then delete the edge $\{a,b\}$;
\item Otherwise, reject $(G,P)$.
\end{enumerate}
\end{redrule}

\begin{proof}[Proof of Safeness]
The safeness of the first three rules is straightforward. As for the fourth rule, let us consider the conditions for when it is applied. In particular, the fourth rule is only called if either $a$ or $b$ occurs in three terminal pairs, or if $a$ occurs in at least one terminal pair and $b$ in at least two but $\{a,b\}\not\in P$. Clearly, $(G,P)$ is a \textbf{NO}-instance in either of these cases.
\end{proof}

The second reduction rule will allow us to replace thin nodes with data tables by small representatives; these representatives will only contain vertices of degree at most $2$ adjacent to the original neighborhood of the thin node. The safeness of this rule follows directly from the definition of $D(t)$ (one simply needs to check each case separately) and hence we do not prove it.

\begin{redrule}
\label{red:thinred}
Let $t$ be a thin node with non-empty $D(t)$.
\begin{enumerate}
\item If $E_t=\{\{a,b\}\}$ where $a\in Y_t$ and if
\begin{itemize}
\item $((\{a,b\}\mapsto L'),\emptyset,\emptyset,\{s,\{a,b\}\})\in D(t)$, then delete $Y_t\setminus \{s\}$ and create the edge $\{s,b\}$; 
\item otherwise, $((\{a,b\}\mapsto U'),\emptyset,\emptyset,\emptyset)\in D(t)$ and we delete $Y_t$.
\end{itemize}

\item If $E_t=\{\{a,b\},\{c,d\}\}$ where $a,c\in Y_t$, $U_t=\emptyset$ and if
\begin{itemize}
\item $((\{a,b\},\{c,d\}\mapsto F'),\emptyset,\{\{a,b\},\{c,d\}\},\emptyset)\in D(t)$, then delete $Y_t$ and create a new vertex $v$ adjacent to $b$ and $d$; else, if
\item $((\{a,b\},\{c,d\}\mapsto U'),\emptyset,\emptyset,\emptyset)\in D(t)$, then delete $Y_t$;
\item otherwise, $((\{a,b\},\{c,d\}\mapsto I'),\{\{a,b\},\{c,d\}\},\emptyset,\emptyset)\in D(t)$ and we delete $Y_t\setminus \{a,c\}$ and add $\{a,c\}$ into $P$.
\end{itemize}

\item If $E_t=\{\{a,b\},\{c,d\}\}$ where $a,c\in Y_t$, $U_t=\{s\}$ and if
\begin{itemize}
\item $((\{a,b\}\mapsto L', \{c,d\}\mapsto U'),\emptyset,\emptyset,\{s,\{a,b\}\})\in D(t)$ and also $((\{c,d\}\mapsto L', \{a,b\}\mapsto U'),\emptyset,\emptyset,\{s,\{c,d\}\})\in D(t)$, then delete $Y_t\setminus \{s\}$ and make $s$ adjacent to $b$ and $d$;
\item otherwise, $((\{a,b\}\mapsto L', \{c,d\}\mapsto U'),\emptyset,\emptyset,\{s,\{a,b\}\})\in D(t)$ and we delete $Y_t\setminus \{s\}$ and make $s$ adjacent to $b$.
\end{itemize}

\item If $E_t=\{\{a,b\},\{c,d\}\}$ where $a,c\in Y_t$, $U_t=\{s_1,s_2\}$ (not necessarily $s_1\neq s_2$) and if
\begin{itemize}
\item $((\{a,b\},\{c,d\}\mapsto L'),\emptyset,\emptyset,\{\{s_1,\{a,b\}\},\{s_2,\{c,d\}\}\})\in D(t)$ and also $((\{a,b\},\{c,d\}\mapsto L'),\emptyset,\emptyset,\{\{s_2,\{a,b\}\},\{s_1,\{c,d\}\}\})\in D(t)$, then add a new vertex $s'$ adjacent to $b$ and $d$, replace all occurrences of $s_1$ and $s_2$ in $P$ by $s'$, and delete $Y_t$;
\item otherwise, $((\{a,b\},\{c,d\}\mapsto L'),\emptyset,\emptyset,\{\{s_1,\{a,b\}\},\{s_2,\{c,d\}\}\})\in D(t)$ and we delete $Y_t\setminus \{s_1,s_2\}$, and make $s_1$ adjacent to $b$ and $s_2$ adjacent to $d$.
\end{itemize}

\item Otherwise, $(G,P)$ is a \textbf{NO}-instance.
\end{enumerate}
\end{redrule}

With Lemma~\ref{lem:simplealg} and Reduction Rules~\ref{red:degtwored},~\ref{red:thinred} in hand, we have all we need to handle the dynamic step. It will be useful to recall the definitions of $A_t$ and $B_t$, and that $|A_t|\leq 2k+1$. 

\begin{proof}[Proof of Lemma~\ref{lem:tcwdyn}]
We begin by looping through all of the at most $4^k\cdot k!$ distinct records for $t$; for each such record $\lambda$, our task is to decide whether it is valid, i.e., whether $(G^\lambda,P^\lambda)$ is a \textbf{YES}-instance. On an intuitive level, our aim will now be to use branching and simplification in order to reduce the question of checking whether $\lambda$ is valid to an instance of \textsc{Simple EDP}.

In our first layer of branching, we will select a record from the data tables of each node in $A_t$. Formally, we say that a \emph{record-set} is a mapping $\tau:t'\in A_t\mapsto \lambda_{t'}\in D(t')$. Note that the number of record-sets is upper-bounded by $(4^k\cdot k!)^{3(2k+1)}$, and we will loop over all possible record-sets.

Next, for each record-set $\tau$, we will apply simplification to each node $t'\in A_t$ in accordance with $\tau(t')$, and recall that each vertex $v$ created by this sequence of simplifications has degree at most $2$. Next, we exhaustively apply Reduction Rule~\ref{red:degtwored} to ensure that each such $v$ is only adjacent to $(V(G)\setminus Y_t)\cup X_t$. At this point, every vertex contained in a bag $X_{t'}$ for $t'\in A_t$ has degree at most $2$ and is only adjacent to $X_t\cup(V(G)\setminus Y_t)$. 

Finally, we apply Reduction Rule~\ref{red:thinred} to replace each thin node by vertices of degree at most $2$ adjacent to $X_t$. At this point, every vertex in $V(G^\lambda)\setminus X_t$ is of degree at most $2$ and only adjacent to $X_t$, and so $(G^\lambda,P^\lambda)$ is an instance of \textsc{Simple EDP}. All that is left is to invoke Lemma~\ref{lem:simplealg}; if it is a \textbf{YES}-instance then we add $\lambda$ to $D(t)$, and otherwise we do not.

The running time is upper bounded by the branching factor $(4^k\cdot k!)^{3(2k+1)}$ times the time to apply our two reduction rules and the time required to solve the resulting \textsc{Simple EDP instance}. All in all, we obtain a running time of at most $k^{\bigoh(k^2)}\cdot |P|^{\bigoh(k^2)}=(k|P|)^{\bigoh(k^2)}$.

We conclude the proof by arguing correctness. Assume $\lambda$ is a valid record. By Definition~\ref{def:valid}, this implies that $(G^\lambda,P^\lambda)$ admits a solution $S$. For each child $t'\in A_t$, $S$ corresponds to some record $\lambda^S_{t'}$ for $t$; consider now the branch in our algorithm which sets $\tau(t')=\lambda^S_{t'}$. Then by Observation~\ref{obs:simplification} it follows that each simplification carried out by the algorithm preserves the existence of a solution to $(G^\lambda,P^\lambda)$. Since both our reduction rules are safe, the instance of \textsc{Simple EDP} we obtain at the end of this branch must also be a \textbf{YES}-instance.

On the other hand, assume the algorithm adds a record $\lambda$ into $D_t$. This means that the resulting \textsc{Simple EDP} instance was a \textbf{YES}-instance. Then by the safeness of our reduction rules and by the second part of Observation~\ref{obs:simplification}, the instance $(G', P')$ obtained by reversing the reduction rules and simplifications was also a \textbf{YES}-instance; in particular $(G^\lambda,P^\lambda)$ is a \textbf{YES}-instance and so $\lambda$ is a valid record.
\end{proof}
}

\lv{
We now have all the ingredients we need to prove our main result.
}

\begin{theorem}
\label{thm:edptcw}
\textsc{EDP} can be solved in time at most $\bigoh(n^3)+k^{\bigoh(k^2)} n^2+ (k|P|)^{\bigoh(k^2)}n$, where $k$ is the treecut width of the input graph and $n$ is the number of its vertices.
\end{theorem}

\begin{proof}
We begin by invoking Theorem~\ref{thm:computetcdec} to compute a treecut decomposition of $G$ of width at most $2k$ and then converting it into a nice treecut decomposition (this takes time $k^{\bigoh(k^2)} n^2$ and $\bigoh(n^3)$, respectively). Afterwards, we use Lemma~\ref{lem:tcwleaves} to compute $D(t)$ for each leaf of $T$, followed by a recursive leaf-to-root application of Lemma~\ref{lem:tcwdyn}. Once we compute $D(r)$ for the root $r$ of $T$, we output \textbf{YES} if and only if $D(r)=\{(\emptyset, \emptyset, \emptyset, \emptyset)\}$.
\end{proof}

\section{A linear kernel for \textsc{EDP}}
\label{sec:fes}
The goal of this section is to provide a fixed-parameter algorithm for \textsc{EDP} which exploits the structure of the input graph exclusively. While treecut width cannot be used to obtain such an algorithm, here we show that the feedback edge set number can. More specifically, we obtain a linear kernel for \textsc{EDP} parameterized by the feedback edge set number. Our kernel relies on the following two facts:

\begin{ourfact}
\label{fact:fescomp}
A minimum feedback edge set of a graph $G$ can be obtained by deleting the edges of minimum spanning trees of all connected components of $G$, and hence can be computed in time $\bigoh(|E(G)|\cdot \log |V(G)|)$.
\end{ourfact}

\begin{ourfact}[\cite{GargVY97}]
\label{fact:festrees}
\EDP\ can be solved in polynomial time when $G$ is a forest.
\end{ourfact}

For the purposes of this section, it will be useful to assume that each vertex $v\in V(G)$ occurs in at
 most one terminal pair, each vertex in a terminal pair has degree $1$
 in $G$, and each terminal pair is not adjacent to each other.
 Note that for any instance without these properties, we can add a new leaf
 vertex for each terminal, attach it to the original terminal, and replace
 the original terminal in $P$ with the leaf vertex~\cite{ZhouTN00,GanianOrdyniakSridharan17}.

Consider an instance $(G,P)$ of \textsc{EDP} and let $X\subseteq E(G)$ be a minimum feedback edge set $X$. Let $Y$ be the set of all vertices incident to at least one edge from $X$, and let $Q=G-X$. Similarly as before, given a subgraph $H$ of $G$, we say that $H$ contains an \emph{unmatched} terminal $s$ if $\{s,t\}\in P$, $s\in V(H)$ and $t\not \in V(H)$.
We begin with two simple reduction rules which allow us to remove degree $2$ vertices and leaves not containing a terminal.

\lv{\begin{redrule}}
\sv{\begin{redrule}[$\star$]}
\label{red:leaves}
Let $v\in V(G)$ be such that $|N_G(v)|=1$. If $v$ is not a terminal, then delete $v$ from $G$.
\end{redrule}

\lv{
\begin{proof}[Proof of Safeness.]
Since $v$ has degree $1$ and is not a terminal, there cannot exist a solution to $(G,P)$ containing a path which uses $va$.
\end{proof}
}

\lv{\begin{redrule}}
\sv{\begin{redrule}[$\star$]}
\label{red:degtwo}
Let $v,a,b\in V(G)$ be such that $N_G(v)=\{a,b\}$ and $\{a,b\}\not \in E$. Then delete $v$ and add the edge $ab$ into $E$. Furthermore, if $\{a,v\}$ or $\{v,b\}$ were in $X$ then add $\{a,b\}$ in $X$.
\end{redrule}

\lv{
\begin{proof}[Proof of Safeness.]
Observe that any solution to the original instance which uses any edge incident to $v$ must contain a path which traverses through both $av$ and $vb$, and after the reduction rule is applied one can simply replace these two edges in that path by $ab$. Any solution in the reduced instance can be similarly transformed into a solution to the original instance. The same argument also shows that the newly constructed set $X$ is also a feedback edge set in the reduced graph.
\end{proof}
}

Of crucial importance is our third rule, which allows us to prune the instance of subtrees with a single edge to $Y$.
For a subgraph $H$ of $G$, recall that $P_H\subseteq P$ denotes the subset of terminal pairs whose both endpoints lie in $H$.

\lv{\begin{redrule}}
\sv{\begin{redrule}[$\star$]}
\label{red:prune}
Let $L$ be a connected component of $G-Y$ such that there exists a single edge $\{\ell\in L, y\in Y\}$ between $L$ and $Y$.
\begin{enumerate}
\item[a.] If $L$ contains no unmatched terminal and $(L,P_L)$ is a \textbf{YES}-instance of \textsc{EDP}, then set $P:=P\setminus P_L$ and $G:=G\setminus V(L)$.
\item[b.] If $L$ contains precisely one unmatched terminal $s$ where $\{s,t\}\in P$ and the instance $(L,P_L\cup\{s,\ell\})$ is a \textbf{YES}-instance of \textsc{EDP}, then set $P:=P\setminus P_L$ and $G:=((V(G)\setminus V(L))\cup \{s\},(E(G)\setminus (E(L)\cup \{\{\ell, y\}\})\cup\{y,s\})$.
\item[c.] In all other cases, $(G,P)$ is a \textbf{NO}-instance of \textsc{EDP}.
\end{enumerate}
\end{redrule}

\lv{
\begin{proof}[Proof of Safeness.]
First, we note that a solution can only use the edge $\{\ell,y\}$ (which is the only edge connecting $L$ to $G-L$) for a path connecting an unmatched terminal in $L$. Let us start by arguing the correctness of Point c., which covers about the following three cases. If $L$ contains at least two distinct unmatched terminals, then any solution to $(G,P)$ would require two edge disjoint paths between $L$ and $G-L$ (which, however, do not exist in $G$). If $L$ contains one unmatched terminal and $(L,P_L\cup\{s,\ell\})$ is a \textbf{NO}-instance of \textsc{EDP}, then it is not possible to find edge disjoint paths which connect terminal pairs in $P_L$ while also connecting $s$ to $t$ (since every path to $t$ must go through $\ell$). Similarly, if $L$ contains zero unmatched terminals and $(L,P_L)$ is a \textbf{NO}-instance of \textsc{EDP}, then $(G,P)$ must also be a \textbf{NO}-instance of \textsc{EDP}.

Next, assume that the conditions of Point a. hold. Since $L$ contains no unmatched terminal, the edge $\{\ell,y\}$ can never be used by any solution and hence it can be removed. Naturally, this results in $L$ being disconnected from $G-L$ and so it suffices to solve $G-L$.

Finally, in the case covered by Point b., every solution to $(G,P)$ must use the edge $\{\ell,y\}$ for an edge disjoint path connecting $s$ to $t$. Hence any solution to the reduced instance in this case implies a solution to $(G,P)$, and similarly any solution to $(G,P)$ can be used to obtain a solution for the reduced instance.
\end{proof}
}

After exhaustive application of Reduction Rules~\ref{red:leaves},~\ref{red:degtwo} and~\ref{red:prune} we observe that each leaf in $Q$ is either in $Y$ or adjacent to a vertex in $Y$. The simple rule below is required to obtain a bound on the number of leaves in $Q$ in the subsequent step.

\lv{\begin{redrule}}
\sv{\begin{redrule}[$\star$]}
\label{red:simplecon}
If $\{a,b\}\in P$ and $a,b$ are leaves in $G$ such that $N(a)=N(b)$, then remove $a$ and $b$ from $G$ and $P$.
\end{redrule}

\lv{
\begin{proof}[Proof of Safeness.]
Every $a$-$b$ path must use the two unique edges incident to $a$ and $b$, and we may assume without loss of generality that a path never visits a vertex twice.
\end{proof}
}

After exhaustive application of Reduction Rules~\ref{red:leaves},~\ref{red:degtwo},~\ref{red:prune} and~\ref{red:simplecon}, we can prove:

\lv{\begin{lemma}}
\sv{\begin{lemma}[$\star$]}
\label{lem:prune}
If $Q$ contains more than $4|X|$ leaves, then $(G,P)$ is a \textbf{NO}-instance.
\end{lemma}

\lv{
\begin{proof}
Let $Z$ denote the set of leaves in $Q$, and let $Y'=Z\cap Y$. Since $|Y|\leq 2|X|$ it follows that $|Z|>|Y|+2|X|$ and hence $|Z\setminus Y'|>|Y\setminus Y'|+2|X|$. For brevity, let $A=Z\setminus Y'$ be the set of leaves excluding the endpoints of our feedback edge set and $B=Y\setminus Y'$ be the set of endpoints of our feedback edge set which are not leaves. 

Recall that every vertex in $A$ is adjacent to a vertex in $B$. Let us define a weight function $w(b)$ for each vertex $b\in B$ which is equal to the number of edges in $X$ incident to $b$; note that since each edge in $X$ contributes by adding at most $2$ to the weight function, we have $\sum_{b\in B}w(b)\leq 2|X|$. Since $|A|>|B|+2|X|\geq |B|+\sum_{b\in B}w(b)=\sum_{b\in B}(w(b)+1)$, there must exist at least one vertex $c\in B$ such that $w(c)+1$ is smaller than the number of its neighbors in $A$; in other words, $c$ is adjacent to at least $w(c)+2$ leaves but is only incident to $w(c)+1$ edges whose endpoints are not leaves. Since $c$ itself is not a leaf and hence not a terminal, every leaf contains a terminal, and Reduction Rule~\ref{red:simplecon} cannot be applied, we conclude that is it not possible to route all terminals located in the leaves adjacent to $c$ to their endpoints via edge disjoint paths.
\end{proof}
}

Finally, we put everything together in the proof of the desired theorem.

\begin{theorem}
\textsc{EDP} admits a linear kernel parameterized by the feedback edge set number of the input graph.
\end{theorem}

\begin{proof}
Let $(G,P)$ be an instance of \textsc{EDP}; w.l.o.g.\ we assume that $G$ is and remains connected (note that if $G$ becomes disconnected due to a later application of a reduction rule, one can simply kernelize each connected component separately). We begin by computing a minimum feedback edge set $X$ of $G$ using Fact~\ref{fact:fescomp}. We then exhaustively apply Reduction Rules ~\ref{red:leaves},~\ref{red:degtwo},~\ref{red:prune} and~\ref{red:simplecon}; since \textsc{EDP} is polynomial-time tractable by Fact~\ref{fact:festrees}, the time required to apply each rule is easily seen to be polynomial. 

After no more rules can be applied, we compare the number of leaves in $Q=G-X$ to $|X|$. If $Q$ contains more than $4|X|$ leaves, then we reject in view of Lemma~\ref{lem:prune}. On the other hand, if $Q$ contains at most $4|X|$ leaves, then we claim that $Q$ contains at most $11|X|-2$ vertices. Indeed, the number of vertices of degree at least $3$ in a forest is at most equal to the number of leaves minus two and in particular $Q$ has at most $4|X|-2$ vertices of degree at least $3$. Moreover, due to the exhaustive application of Reduction Rule~\ref{red:degtwo} it follows that the number of degree two vertices is at most $|X|$. And so, by putting together the bounds on $|Y|$ along with the number of vertices of degree $1$ and $2$ and $3$, we obtain $|V(G)|=|V(Q)|\leq 2|X|+4|X|+|X|+4|X|-2$, as claimed.
\end{proof}

%
%

\bibliographystyle{plainurl}
\bibliography{literature}

\begin{thebibliography}{10}

\bibitem{BerendTassa10}
Daniel Berend and Tamir Tassa.
\newblock Improved bounds on bell numbers and on moments of sums of random
  variables.
\newblock {\em Probability and Mathematical Statistics}, 30(2):185--205, 2010.

\bibitem{ChekuriKS06}
Chandra Chekuri, Sanjeev Khanna, and F.~Bruce Shepherd.
\newblock An {O}(sqrt(n)) approximation and integrality gap for disjoint paths
  and unsplittable flow.
\newblock {\em Theory of Computing}, 2(7):137--146, 2006.

\bibitem{CourcelleMR00}
Bruno Courcelle, Johann~A. Makowsky, and Udi Rotics.
\newblock Linear time solvable optimization problems on graphs of bounded
  clique-width.
\newblock {\em Theory Comput. Syst.}, 33(2):125--150, 2000.

\bibitem{CyganFKLMPPS15}
Marek Cygan, Fedor~V. Fomin, Lukasz Kowalik, Daniel Lokshtanov, D{\'{a}}niel
  Marx, Marcin Pilipczuk, Michal Pilipczuk, and Saket Saurabh.
\newblock {\em Parameterized Algorithms}.
\newblock Springer, 2015.

\bibitem{CyganFKLMPPS14}
Marek Cygan, Fedor~V. Fomin, {\L}ukasz Kowalik, Daniel Lokshtanov, D\'aniel
  Marx, Marcin Pilipczuk, Micha{\l} Pilipczuk, and Saket Saurabh.
\newblock {\em Parameterized Algorithms}.
\newblock Springer-Verlag, Berlin, to appear in 2014.

\bibitem{Diestel10}
Reinhard Diestel.
\newblock {\em Graph Theory}.
\newblock Springer-Verlag, Heidelberg, 4th edition, 2010.

\bibitem{DowneyFellows13}
Rodney~G. Downey and Michael~R. Fellows.
\newblock {\em Fundamentals of Parameterized Complexity}.
\newblock Texts in Computer Science. Springer, 2013.

\bibitem{EneMPR16}
Alina Ene, Matthias Mnich, Marcin Pilipczuk, and Andrej Risteski.
\newblock On routing disjoint paths in bounded treewidth graphs.
\newblock In {\em Proc. {SWAT} 2016}, volume~53 of {\em LIPIcs}, pages
  15:1--15:15. Schloss Dagstuhl, 2016.

\bibitem{FleszarMS16}
Krzysztof Fleszar, Matthias Mnich, and Joachim Spoerhase.
\newblock New algorithms for maximum disjoint paths based on tree-likeness.
\newblock In {\em Proc. {ESA} 2016}, pages 42:1--42:17, 2016.

\bibitem{FlumGrohe06}
J\"{o}rg Flum and Martin Grohe.
\newblock {\em Parameterized Complexity Theory}, volume XIV of {\em Texts in
  Theoretical Computer Science. An EATCS Series}.
\newblock Springer Verlag, Berlin, 2006.

\bibitem{GanianKimSzeider15}
Robert Ganian, Eun~Jung Kim, and Stefan Szeider.
\newblock Algorithmic applications of tree-cut width.
\newblock In Giuseppe~F. Italiano, Giovanni Pighizzini, and Donald Sannella,
  editors, {\em Proc. {MFCS} 2015}, volume 9235 of {\em LNCS}, pages 348--360.
  Springer, 2015.

\bibitem{GanianKO18}
Robert Ganian, Fabian Klute, and Sebastian Ordyniak.
\newblock On structural parameterizations of the bounded-degree vertex deletion
  problem.
\newblock In {\em 35th Symposium on Theoretical Aspects of Computer Science,
  {STACS} 2018, February 28 to March 3, 2018, Caen, France}, pages 33:1--33:14,
  2018.

\bibitem{GanianOS17}
Robert Ganian, Sebastian Ordyniak, and Ramanujan Sridharan.
\newblock On structural parameterizations of the edge disjoint paths problem.
\newblock In {\em 28th International Symposium on Algorithms and Computation,
  {ISAAC} 2017, December 9-12, 2017, Phuket, Thailand}, pages 36:1--36:13,
  2017.

\bibitem{GanianOrdyniakSridharan17}
Robert Ganian, Sebastian Ordyniak, and Ramanujan Sridharan.
\newblock On structural parameterizations of the edge disjoint paths problem.
\newblock In Yoshio Okamoto and Takeshi Tokuyama, editors, {\em 28th
  International Symposium on Algorithms and Computation, {ISAAC} 2017, December
  9-12, 2017, Phuket, Thailand}, volume~92 of {\em LIPIcs}, pages 36:1--36:13.
  Schloss Dagstuhl - Leibniz-Zentrum fuer Informatik, 2017.

\bibitem{GargVY97}
Naveen Garg, Vijay~V. Vazirani, and Mihalis Yannakakis.
\newblock Primal-dual approximation algorithms for integral flow and multicut
  in trees.
\newblock {\em Algorithmica}, 18(1):3--20, 1997.

\bibitem{Karp75}
Richard~M Karp.
\newblock On the computational complexity of combinatorial problems.
\newblock {\em Networks}, 5(1):45--68, 1975.

\bibitem{KawarabayashiKK14}
Ken{-}ichi Kawarabayashi, Yusuke Kobayashi, and Stephan Kreutzer.
\newblock An excluded half-integral grid theorem for digraphs and the directed
  disjoint paths problem.
\newblock In {\em Proc. {STOC} 2014}, pages 70--78. {ACM}, 2014.

\bibitem{KimOPST15}
Eunjung Kim, Sang{-}il Oum, Christophe Paul, Ignasi Sau, and Dimitrios~M.
  Thilikos.
\newblock An {FPT} 2-approximation for tree-cut decomposition.
\newblock In Laura Sanit{\`{a}} and Martin Skutella, editors, {\em Proc. {WAOA}
  2015}, volume 9499 of {\em LNCS}, pages 35--46. Springer, 2015.

\bibitem{KolliopoulosS04}
Stavros~G. Kolliopoulos and Clifford Stein.
\newblock Approximating disjoint-path problems using packing integer programs.
\newblock {\em Math. Program.}, 99(1):63--87, 2004.

\bibitem{MarxWollan14}
D{\'{a}}niel Marx and Paul Wollan.
\newblock Immersions in highly edge connected graphs.
\newblock {\em {SIAM} J. Discrete Math.}, 28(1):503--520, 2014.

\bibitem{NishizekiVZ01}
Takao Nishizeki, Jens Vygen, and Xiao Zhou.
\newblock The edge-disjoint paths problem is {NP}-complete for series-parallel
  graphs.
\newblock {\em Discrete Applied Mathematics}, 115(1-3):177--186, 2001.

\bibitem{RobertsonS95b}
Neil Robertson and Paul~D. Seymour.
\newblock Graph minors {XIII}. {T}he disjoint paths problem.
\newblock {\em J. Comb. Theory, Ser. {B}}, 63(1):65--110, 1995.

\bibitem{RobertsonS03b}
Neil Robertson and Paul~D. Seymour.
\newblock Graph minors. {XVIII.} tree-decompositions and well-quasi-ordering.
\newblock {\em J. Comb. Theory, Ser. {B}}, 89(1):77--108, 2003.

\bibitem{Scheffler94}
Petra Scheffler.
\newblock Practical linear time algorithm for disjoint paths in graphs with
  bounded tree-width.
\newblock In {\em Technical Report TR 396/1994}. FU Berlin, Fachbereich 3
  Mathematik, 1994.

\bibitem{Wollan15}
Paul Wollan.
\newblock The structure of graphs not admitting a fixed immersion.
\newblock {\em J. Comb. Theory, Ser. {B}}, 110:47--66, 2015.

\bibitem{ZhouTN00}
Xiao Zhou, Syurei Tamura, and Takao Nishizeki.
\newblock Finding edge-disjoint paths in partial \emph{k}-trees.
\newblock {\em Algorithmica}, 26(1):3--30, 2000.

\end{thebibliography}

\end{document}